\documentclass[journal,final]{IEEEtran}

\usepackage{amsmath,amssymb,amsthm}
\usepackage{xifthen}
\usepackage{mathtools}
\usepackage{enumerate}
\usepackage{microtype}
\usepackage{xspace}
\usepackage{bm}
\usepackage[T1]{fontenc}
\usepackage{fancyhdr}
\usepackage{lastpage}
\usepackage{bbm}

\newcommand{\mbs}[1]{\pmb{#1}}
\newcommand{\vect}[1]{{\lowercase{\mbs{#1}}}}
\newcommand{\mat}[1]{{\uppercase{\mbs{#1}}}}

\newcommand{\T}{{\scriptscriptstyle\mathsf{T}}}
\renewcommand{\H}{{\scriptscriptstyle\mathsf{H}}}

\newcommand{\cond}{\,\vert\,}
\renewcommand{\Re}[1][]{\ifthenelse{\isempty{#1}}{\operatorname{Re}}{\operatorname{Re}\left(#1\right)}}
\renewcommand{\Im}[1][]{\ifthenelse{\isempty{#1}}{\operatorname{Im}}{\operatorname{Im}\left(#1\right)}}

\newcommand{\hv}{\vect{h}}

\newcommand{\rv}{\vect{r}}

\newcommand{\vv}{\vect{v}}

\newcommand{\xv}{\vect{x}}

\newcommand{\zerov}{\vect{0}}

\newcommand{\Hm}{\mat{h}}

\newcommand{\Mm}{\mat{M}}

\newcommand{\Qm}{\mat{q}}

\newcommand{\Wm}{\mat{w}}

\newcommand{\Ic}{{\mathcal I}}

\newcommand{\Kc}{{\mathcal K}}

\newcommand{\Sc}{{\mathcal S}}

\newcommand{\Id}{\mat{\mathrm{I}}}

\newcommand{\CN}[1][]{\ifthenelse{\isempty{#1}}{\mathcal{N}_{\mathbb{C}}}{\mathcal{N}_{\mathbb{C}}\left(#1\right)}}

\renewcommand{\P}[1][]{\ifthenelse{\isempty{#1}}{\mathbb{P}}{\mathbb{P}\left(#1\right)}}
\newcommand{\E}[1][]{\ifthenelse{\isempty{#1}}{\mathbb{E}}{\mathbb{E}\left[#1\right]}}
\newcommand{\I}[1][]{\ifthenelse{\isempty{#1}}{\mathbb{I}}{\mathbb{I}\left\{#1\right\}}}
\renewcommand{\det}[1][]{\ifthenelse{\isempty{#1}}{\mathrm{det}}{\text{det}\left(#1\right)}}
\newcommand{\trace}[1][]{\ifthenelse{\isempty{#1}}{\mathrm{tr}}{\text{tr}\left(#1\right)}}
\newcommand{\rank}[1][]{\ifthenelse{\isempty{#1}}{\mathrm{rank}}{\text{rank}\left(#1\right)}}
\newcommand{\diag}[1][]{\ifthenelse{\isempty{#1}}{\mathrm{diag}}{\text{diag}\left(#1\right)}}
\newcommand{\Cov}[1][]{\ifthenelse{\isempty{#1}}{\mathsf{Cov}}{\mathsf{Cov}\left(#1\right)}}

\newtheorem{proposition}{Proposition}
\newtheorem{remark}{Remark}
\newtheorem{example}{Example}
\newtheorem{corollary}{Corollary}
\newtheorem{lemma}{Lemma}

\renewcommand{\rv}[1]{{\mathrm{#1}}}
\newcommand{\rvVec}[1]{\pmb{\mathrm{#1}}}

\newcounter{enumi_saved}
\usepackage{answers}
\Newassociation{solution}{Solution}{solutionfile}

\AtBeginDocument{\Opensolutionfile{solutionfile}[\jobname]}
\AtEndDocument{\Closesolutionfile{solutionfile}\clearpage
}

\IfFileExists{MinionPro.sty}{
}{
}

\usepackage{subfigure}
\usepackage{fancyhdr}
\usepackage{lastpage}
\usepackage{bbm}
\usepackage{graphicx}
\usepackage{array,multirow}
\usepackage{algorithm, algorithmicx, algpseudocode}

\usepackage{mathtools}
\usepackage{stfloats}
\usepackage{lipsum}
\usepackage{multirow}

\renewcommand{\rv}[1]{{\mathrm{#1}}}

\newcommand{\Conv}{\mathsf{Conv}}
\newcommand{\nt}{n_{\text{t}}}
\newcommand{\nr}{n_{\text{r}}}
\allowdisplaybreaks

\usepackage{color}

\newtheorem{claim}{Claim}
\newcommand{\collectionS}{\pmb{\Sc}}

\newcommand{\CG}[1]{\tilde{#1}}

\title{On Linearly Precoded Rate Splitting 
for {Gaussian} MIMO Broadcast Channels}
\author{Zheng Li,~\IEEEmembership{Member,~IEEE,} Sheng~Yang,~\IEEEmembership{Member,~IEEE,}
and\\ Shlomo Shamai~(Shitz),~\IEEEmembership{Fellow,~IEEE}
\thanks{Z. Li is with Orange Labs
         Networks, 92326, Ch\^atillon, France. S. Yang is with the laboratory of signals and systems at
         CentraleSup\'elec, Paris-Saclay University, 91190,
         Gif-sur-Yvette, France.   
         S.~Shamai~(Shitz) is with Technion-Israel Institute of
         Technology, Haifa, Israel. Email:zheng1.li@orange.com,\ 
         sheng.yang@centralesupelec.fr,\ sshlomo@ee.technion.ac.il.}
         \thanks{This paper was presented in part at IEEE SPAWC 2018 and
         IEEE ITW 2018. }
        \thanks{The work of S. Shamai has been supported by the
European Union's Horizon 2020 Research And Innovation Programme,
grant agreement no.~694630.} 
\thanks{Copyright (c) 2017 IEEE. Personal use of this material is
permitted.  However, permission to use this material for any other
purposes must be obtained from the IEEE by sending a request to
pubs-permissions@ieee.org.}
         }

\begin{document}

\maketitle

\begin{abstract}
In this paper, we consider a general $K$-user {Gaussian} multiple-input
multiple-output~(MIMO) broadcast channel~(BC). We assume that the
channel state is deterministic and known to all the nodes. While the
{private-message} capacity region is well known to be achievable with dirty paper coding
(DPC), we are interested in the simpler linearly precoded transmission
schemes. {In particular, we focus on linear precoding
schemes combined with rate-splitting~(RS). First, we derive an achievable
rate region with minimum mean square error~(MMSE) precoding at the
transmitter and joint decoding of the sub-messages at the receivers. 
Then, we study the achievable sum rate of this scheme and obtain two
findings: 1)~an analytically tractable upper bound on the sum rate that
is shown numerically to be a close approximation, and 2)~how to reduce the number of active
streams -- crucial to the overall complexity -- while preserving the sum rate to
within a constant loss. The latter results in two practical algorithms:
a stream elimination algorithm and a stream ordering algorithm. 
Finally, we investigate the constant-gap optimality of linearly precoded
RS with respect to the capacity. Our result reveals that,
while the achievable rate of linear precoding alone can be arbitrarily
far from the capacity, the introduction of RS can help
achieve the capacity region to within a constant gap in the two-user case.
Nevertheless, we prove that the RS scheme's constant-gap optimality does not extend to the three-user case. 
Specifically, we show, through a pathological example, that the gap
between the sum rate and the sum capacity can be unbounded. 
}
 
\end{abstract}


\section{Introduction}

The capacity region of a multi-antenna~(MIMO) broadcast
channel~(BC) with additive Gaussian noise has been
characterized for more than a decade~\cite{Caire:2003, MIMOBC2006}.
The capacity achieving scheme is essentially the dirty paper coding~(DPC)~\cite{costa1983writing} combined
with the minimum mean square error~(MMSE) precoding. This BC capacity region can be conveniently represented
with the capacity region of the dual multiple-access channel~(MAC)
via the so-called MAC-BC duality~(also known as the uplink-downlink
duality)~\cite{viswanath2003sum, jindal2004duality}.
The main role of the DPC can be regarded as \emph{interference mitigation at
the transmitter side}, i.e., part of the interference is pre-cancelled
for a given receiver at the transmitter side.
The implementation of DPC is however not trivial, due to
its non-linear nature and the fact that it is sensitive to the channel state information at
the transmitter side~(CSIT)~\cite{yang2005impact}. As such, linear precoding is used in most
practical systems instead. Apart from the low implementation complexity, it can be shown that linear precoding
schemes such as zero-forcing~(ZF) achieve the maximum degrees of freedom~(DoF) of the
system~\cite{yoo2006optimality,lee2007high}. Intuitively, ZF is 
sufficient for the transmitter to exploit all the available
dimensions of the signal space in a BC, leading consequently to the DoF optimality. 

Despite its simplicity, the dimension-counting DoF metric is coarse
since it only characterizes the pre-log factor of the {achievable rate}
when the channel gains are bounded while the signal-to-noise ratio~(SNR)
goes to infinity. As a result, it fails to capture the disparity of the
channel strengths among users, and thus {in some cases} provides little information on
the system behavior for different channel realizations. {To see how ZF
can be useless and the DoF metric can be meaningless {for some
channel realizations}, let us consider
the following toy example with two users. Let the channel vectors from
the transmitter to the receivers be $[\begin{matrix}\sqrt{1-\epsilon^2}&
  \epsilon \end{matrix}]$ and $[\begin{matrix}\sqrt{1-\epsilon^2}&
    -\epsilon \end{matrix}]$, respectively, that are linearly
    independent for any $\epsilon\in(0,1)$. With ZF, the beam directions
    for the receivers would be $[\begin{matrix}\epsilon &
      \sqrt{1-\epsilon^2} \end{matrix}]$ and $[\begin{matrix}-\epsilon &
        \sqrt{1-\epsilon^2} \end{matrix}]$, respectively, both
        perpendicular to the other receiver's channel to avoid
        interference. Provided that each stream has power $P/2$ and the
        noise power at each receiver is $1$, the achievable rate for
        each user is $\log(1+2\epsilon^2(1-\epsilon^2) P)$. Note that
        the DoF analysis would completely erase the impact of any
        non-zero $\epsilon$ and give $1$ DoF for each user, while the
        actual rate can be arbitrarily close to $0$ when $\epsilon$ is
        close to $0$ or to $1$. In fact, serving only one user would
        provide a rate $\log(1+P)$, much larger than the sum rate of ZF
        in those extreme cases. Indeed, the two receivers' signal spaces
        can have a non-negligible overlap so that nullifying
        interference at the transmitter~(e.g., ZF) could be highly
        suboptimal. }

 To account for the relative strength of the channel
coefficients, one can let the channel gains of different links grow with
the SNR polynomially with different exponents, and the resulting
pre-log {of the achievable rate} is called the generalized DoF~(GDoF). 
{An even finer characterization is the \emph{constant-gap rate}, a rate
that is within a constant gap to the exact achievable rate for
\emph{any} channel realization. Indeed, since the constant-gap rate and the
exact rate have the same pre-log when the SNR goes to infinity, the GDoF
and DoF can also be derived from the constant-gap rate.}
Therefore, we have the following progressive improvements on
the {rate} approximation~\cite{davoodi2017transmitter}: 
$$ \text{DoF}\ll\text{GDoF}\ll\text{constant-gap}.$$

To circumvent the limitation of interference, we can introduce rate
splitting~(RS) so that interference is \emph{decodable}. 
The idea of using RS to partially mitigate interference was first
proposed for the two-user interference
channels~\cite{carleial1978interference, HK:1981}, in which independent messages are sent by independent
transmitters to their respective receivers.
Essentially, each individual message is split into one private part and one common
part, where the common part is decodable by~(though not intended to)
both receivers. Each receiver decodes and thus can remove the common message from the
interfering transmitter. It turned out that
such a scheme achieves the capacity region of the two-user interference
channel to within 1~{bit/s/Hz}~\cite{etkin2008gaussian}. The same idea can also be applied
to the BCs. In \cite{yang2013degrees}, the authors showed
that RS can provide a strict sum DoF gain of a BC when
only imperfect CSIT is available. 
Extensions to different settings have
been made in later works \cite{davoodi2018gdof, mao2018rate, piovano2017optimal}. Besides, RS has also been considered for robust transmissions under bounded CSIT errors in \cite{joudeh2016robust}.
In contrast to the DPC that pre-cancels interference at the transmitter
side, RS enables the \emph{interference
mitigation at the receivers' side} by letting the interference  
decodable by the receivers. 

{In this work, we are interested in the constant-gap rate of linearly
precoded RS schemes for {Gaussian} MIMO BC.
In particular,} we consider a general RS scheme with MMSE precoding and
{joint decoding of the sub-messages at the receivers}, and
characterize the corresponding achievable rate region. {The main
contributions of our work are summarized as follows. 
\begin{itemize}
  \item A major challenge in investigating the general $K$-user case is
    the large number of rate constraints, up to $K 2^{2^{K-1}}$ in
    general. {Characterizing the maximum sum rate, even up to a constant
    gap, is hard when $K$ is only moderately large. For instance, there
    are about $3\times10^5$ and $2\times 10^{10}$ constraints when $K=5$
    and $K=6$, respectively.} Our contribution here 
    is to analyze the achievable constant-gap sum rate and obtain two meaningful results. 
    First, while the set of rate constraints is large, we
    can carefully choose a subset that leads to $K$ closed-form upper bounds. 
    Remarkably, the proposed upper bound, as the minimum
    value of the aforementioned $K$ upper bounds, turns out to be a
    numerically close approximation. 
    Second, we show how to reduce the number of active streams, which
    is crucial to the overall complexity, while preserving the
    constant-gap sum rate. Specifically, we propose two practical
    algorithms: 1)~a stream elimination algorithm that removes streams
    without causing a rate loss more than a given target value, and 2)~a stream
    ordering algorithm that orders the entire set of $2^K-1$ streams
    according to their impact to the sum rate. 
  \item A central theoretical question on any communication scheme is
    whether it achieves the channel capacity or, if not, how far it is
    from capacity-achieving. Our contribution here is to investigate
    whether linearly precoded RS is constant-gap optimal, that is,
    achieves the capacity to within a constant gap~(constant-gap
    capacity in short). After showing that any linear precoding scheme alone
    cannot be constant-gap optimal, we go on and prove that with
    rate-splitting the entire achievable constant-gap rate region coincides with the capacity
    region in the two-user case. Nevertheless, we
    show that such optimality does not extend beyond two users.
    Specifically, we use the derived sum rate upper bound and a simple pathological
    channel realization to demonstrate an unbounded gap to
    the sum capacity in the three-user case. {In fact, the RS scheme is not even GDoF optimal in
    this example. We argue that without a proper codebook design for
    interference alignment or interference pre-cancellation~(e.g., DPC),
    the independent interference streams become overwhelming for each
    individual receiver to decode. Our study thus reveals a fundamental
    gap between the receiver-side interference mitigation and the
    transmitter-side interference mitigation.}
\end{itemize}
}

In the literature, there have been quite a few works on RS for BC in
recent years. While some of the works consider perfect CSIT, most of
works apply RS to mitigate interference caused by CSIT imperfection. In
particular, one can consider the GDoF while letting the CSIT error
scales as $\mathsf{SNR}^{-\beta}$, where $\beta\ge0$ is used to measure
the CSIT accuracy~\cite{joudeh2016robust}. In such a setting, the
presence of the common message to all users is necessary to make full
use of the transmit power. It is shown in~\cite{davoodi2017transmitter}
that using only one common message together with the private messages is
optimal in a \emph{symmetric} $K$-user setting. Many other works focus on the
optimization problems related to precoder designs for different channel
models, with perfect or imperfect CSIT assumptions. Some optimize the
sum rate~\cite{mao2018rate}, while others focus on
the pre-log factor at high SNR
\cite{piovano2017optimal,joudeh2016robust,joudeh2016sum}. MMSE precoder
has also been considered in previous works but {is} only limited for private
messages \cite{lu2017mmse,dai2016rate}.  In this work, we are interested
in the unanswered, yet fundamental, question of whether the RS scheme
can be close to optimal in a stronger sense than the DoF even with
perfect CSIT~(which can be regarded as an extreme case of the imperfect
CSIT). 
{Our work's technical and new contribution beyond the state of the art is the constant-gap analyses of linearly precoded RS schemes.
Furthermore, while most of the existing works on the RS are limited to
one or two layers of common messages~\cite{mao2018rate,dai2016rate}, our work is the first to
characterize an analytical upper bound of the general $K$-user RS
scheme, to the best of our knowledge. This is also the first attempt to
propose a constructive way {with an analytical criterion} to reduce the number of streams according to 
the actual channel realization, in contrast to previous works
{that apply numerical simulations for such purposes. }

The remainder of the paper is organized as follows.  We present the
channel model in Section~\ref{sec:model}. {In
Section~\ref{sec:RS}, we describe the RS scheme in its most general form with MMSE
precoding, {and derive the achievable constant-gap rate region}. 
The constant-gap sum rate is analyzed in Section~\ref{sec:sum-rate}
with the general $K$-user sum rate upper bound
derived in Section~\ref{sec:Kuser} and the stream elimination and
ordering algorithms described in Section~\ref{sec:Algo}. The
constant-gap optimality is then investigated in
Section~\ref{sec:optimality}, before we  
conclude the paper in Section~\ref{sec:conclusion}.}

\section{System Model and Preliminaries} \label{sec:model}
\subsubsection*{Notation}
In this paper, we use the following notational conventions. For random
quantities, we use upper case non-italic letters, e.g., $\rv{X}$, for
scalars, {and} upper case
non-italic bold letters, e.g., $\rvVec{V}$, for vectors. Deterministic quantities are denoted in a rather conventional way
with italic letters, e.g., a scalar $x$, a vector $\pmb{v}$, and a matrix
$\pmb{M}$. Logarithms are in base $2$.  The Euclidean norm of a vector and a
matrix is denoted by $\|\vv\|$ and $\|\Mm\|$, respectively. $\Mm^T$, $\Mm^H$, $\trace(\Mm)$ and $\det(\Mm)$ is the transpose, the conjugate transpose, the trace and the determinant of a matrix $\Mm$, respectively. 
{$\Mm(i,j)$ is the $(i,j)$-th entry of the matrix $\Mm$.} 
$[K]$ is the
set $\{1,\ldots,K\}$, while $[n]$ represents the 
set $\{1,\ldots,n\}$. Subsets are denoted with calligraphic capitalized
letters, e.g., $\mathcal{K}$ and $\mathcal{S}$. $|\mathcal{S}|$ represents the cardinality of the set $\mathcal{S}$. 
We use $2^\mathcal{K}$ to denote the power set of $\mathcal{K}$, i.e.,
the collection of all subsets of $\mathcal{K}$. We use $\bar{\Kc}$ to
denote the complement set of $\Kc$, i.e., $\bar{\Kc} = [K]\setminus
\Kc$ if $[K]$ is the whole set.
To avoid confusion, we use bold calligraphic letters to specify sets of
sets, e.g., $\pmb{\Sc}$, which are referred to as \emph{collections}. 
For convenience, we use
$\{A_k\}_k$ to denote the set $\{A_k:\,k=1,\ldots,K\}$ where $A$ can be
any object.  
The colon equal ``:='' denotes equality by definition or assignment. 
``$\Conv$'' stands for the convex hull operation. 
{$(x)^+ := \max\left\{ x,0 \right\}$.}
Throughout the paper, we use ``$\approx$'' for constant-gap
approximation.

\subsection{Channel model}
We consider a $K$-user time-invariant and frequency-flat {Gaussian} MIMO BC where
the transmitter has $\nt$ antennas. The channel output at receiver~$k$,
$k\in[K]$, at time $t$, $t\in[n]$, is
\begin{align}
  \rvVec{Y}_{\!k}[t] = \Hm_{\!k} \, \xv[t] + \rvVec{Z}_k[t], 
\end{align}%
or, in a compact form
\begin{align}
  \bigl[  \rvVec{Y}_{\!1}^\T[t] \ \cdots\ \rvVec{Y}_{\!K}^\T[t]
   \bigr]^\T &= \Hm \xv[t] +
    \rvVec{Z}[t], 
\end{align}%
where $\rvVec{Z}[t] \sim \mathcal{CN}(0,\Id)$ is the temporally
independent and identically distributed~(i.i.d.) additive white Gaussian
noise~(AWGN) with normalized variance;
$\Hm_{\!k}\in\mathbb{C}^{n_{\text{r},k}\times\nt}$ is the channel matrix from the
transmitter to the receiver~$k$, $n_{\text{r},k}$ being the number of
antennas at receiver~$k$; $\Hm :=
\bigl[\Hm_{\!1}^\T\,\cdots\,\Hm_{\!K}^\T\bigr]^\T$ is the
global channel matrix assumed to be deterministic and is known globally.
The input sequence is subject
to the power constraint $\frac{1}{n}\sum_{t=1}^n \| \xv[t] \|^2 \le P$
where $P$ is identified with the SNR. 

\subsection{Capacity region}\label{sec:capacity_def}

Let us assume that the transmitter sends an independent message to each
receiver~$k$, $k\in[K]$, at a rate $R_k$ {bits/s/Hz}. The capacity region,
denoted by $\mathcal{C}_{\text{BC}}(\{\Hm_{\!k}\}_k, P)$, 
is the set of rate-tuples $(R_1,\ldots,R_K)$ such that the probability
of decoding error can be arbitrarily small when $n\to\infty$.   
This capacity region can be conveniently characterized with the so-called
MAC-BC duality. Namely, the capacity region of a MIMO BC with power
constraint $P$ is the union of the capacity regions of the dual MAC over all individual power constraints that sum to $P$, i.e.,
\begin{IEEEeqnarray}{rCl}
  \mathcal{C}_{\text{BC}}(\{\Hm_{\!k}\}_k, P) =
  \bigcup_{\{\Qm_k\}_k: \sum_{k=1}^K\!\! \trace(\Qm_k)\le P}\!\!\!\!\!\!\!\!\!\!\!
  \mathcal{C}_{\text{MAC}}(\{\Hm_{\!k}^\H\}_k, \{\Qm_k\}_k),
  \IEEEeqnarraynumspace
\end{IEEEeqnarray}%
where the region $\mathcal{C}_{\text{MAC}}(\{\Hm_{\!k}^\H\}_k, \{\Qm_k\}_k)$
denotes the capacity region of the dual MAC under the individual
covariance constraint $\Qm_k$, $k\in[K]$. In fact, it is well
known~{(see, e.g.,~\cite{tse1998multiaccess, NIT})} that
$\mathcal{C}_{\text{MAC}}(\{\Hm_{\!k}^\H\}_k, \{\Qm_k\}_k)$ is a polymatroid
with the set of {non-negative} rate tuples satisfying
\begin{align}
\sum_{k\in\mathcal{K}} R_k \le \log\det\Bigl(\Id + \sum_{k\in\mathcal{K}}
\Hm_{\!k}^\H \Qm_k \Hm_{\!k}\Bigl), \quad \forall\,\mathcal{K}\subseteq [K]. \label{eq:C-MAC}
\end{align}%
Since the log-det function is increasing with the partial ordering of
positive semi-definite matrices, it follows from $\Qm_k\preceq P\,\Id$,
$k\in[K]$, that
\begin{equation}
  \mathcal{C}_{\text{MAC}}\Bigl(\{\Hm_{\!k}^\H\}_k, \frac{P}{\nr}\Id\Bigr)
  \subseteq \mathcal{C}_{\text{BC}}(\{\Hm_{\!k}\}_k, P) \subseteq
  \mathcal{C}_{\text{MAC}}(\{\Hm_{\!k}^\H\}_k, P\Id), \label{eq:mac-bc}
\end{equation}
where $\nr:=\sum_{k=1}^K n_{\text{r},k}$. Hence, we have the following
lemma. 
\begin{lemma}\label{app:lemma1}
  The BC capacity region $\mathcal{C}_{\text{BC}}(\{\Hm_{\!k}\}_k, P)$ is
  within $\gamma:=\nt\log \nr$ {bits/s/Hz} to the MAC capacity region
$\mathcal{C}_{\text{MAC}}(\{\Hm_{\!k}^\H\}_k, P\Id)$, such that, 
\begin{multline}
  {\text{if } (R_1,\ldots,R_K) \in
  \mathcal{C}_{\text{MAC}}(\{\Hm_{\!k}^\H\}_k, P\Id),} \\ \text{then }
  ( (R_1-\gamma)^+,\ldots,(R_K-\gamma)^+)
  \in\mathcal{C}_{\text{BC}}(\{\Hm_{\!k}\}_k,
  P). \nonumber 
\end{multline}
\end{lemma}
\begin{proof}
  From \eqref{eq:mac-bc}, we only need to show that
  $\mathcal{C}_{\text{MAC}}(\{\Hm_{\!k}^\H\}_k, \frac{P}{\nr}\Id)$ is
  within $\gamma$ {bits/s/Hz} to $\mathcal{C}_{\text{MAC}}(\{\Hm_{\!k}^\H\}_k,
P\Id)$. Indeed, since $\log\det(\Id +\sum_{k\in\mathcal{K}}\!
\frac{P}{\nr} \Hm_{\!k}^\H \Hm_{\!k}) {=} \log\det(\nr\Id +\sum_{k\in\mathcal{K}}
{P} \Hm_{\!k}^\H \Hm_{\!k}) - \gamma \!\ge \!\log\det(\Id + \!\sum_{k\in\mathcal{K}}
{P} \Hm_{\!k}^\H \Hm_{\!k}) - \gamma$, the proof is immediate from
the definition of the MAC region in \eqref{eq:C-MAC}. 
\end{proof}

An optimal scheme that achieves the exact capacity region consists in
combining the DPC with the MMSE
precoding~\cite{NIT}. Specifically, for a given encoding order, each message is 
first encoded using Costa's DPC that pre-cancels the previously encoded
signals, and is then precoded with the MMSE matrix. 
In this way, each receiver only sees the interference from the 
messages that are encoded afterwards. Here we can clearly see the duality
between the successive encoding of the BC and the successive interference
cancellation~(SIC) decoding of the MAC. Since part of the interference
is pre-cancelled at the transmitter side and the receivers treat the
residual interference as additive noise, such a scheme can be regarded
as transmitter-side interference mitigation. 
While the MMSE precoding is linear, the DPC is non-linear and can be implemented with
nested lattices~\cite{zamir2002nested}.

\subsection{Constant-gap rate region}\label{sec:def_constant}

In the following, we give a formal definition of the constant-gap rate
region. Let $\mathcal{C}(\Hm,P)$ be the capacity region of a given
channel $\Hm$ with power constraint $P$, and $\mathcal{R}(\Hm,P)$ the achievable rate
region of some scheme. We say that $\CG{\mathcal{R}}(\Hm,P)$ is an
achievable constant-gap rate region with respect to
$\mathcal{R}(\Hm,P)$, if   
\begin{equation}
  \sup_{P\ge0,\Hm}\max_{\CG{\pmb{r}}\in\CG{\mathcal{R}}} \min_{{\pmb{r}}\in{\mathcal{R}}}\|\tilde{\pmb{r}}-\pmb{r}\| < \infty. 
\end{equation}
In words, it means that every rate tuple in the constant-gap rate region
is achievable by the given scheme, to within a constant gap. A scheme is said
to be constant-gap optimal if the capacity region $\mathcal{C}(\Hm,P)$
is an achievable constant-gap rate region with respect to
$\mathcal{R}(\Hm,P)$. $\bar{R}(\Hm,P)$ is said to be a constant-gap sum rate
upper bound of a given scheme if there exists a constant~$\gamma$ such
that any achievable rate by the given scheme is upper bounded by
$\bar{R}(\Hm,P)+\gamma$ for any channel realization $\Hm$ and transmit
power $P$. 

For future reference, we also recall the DoF and GDoF optimalities. 
The scheme is DoF optimal if 
\begin{equation}
  \lim_{P\to\infty}\max_{\tilde{\pmb{r}}\in\mathcal{C}}\min_{\pmb{r}\in{\mathcal{R}}}\frac{\|\tilde{\pmb{r}}-\pmb{r}\|}{\log
  P} = 0,\ \forall\,\Hm.
\end{equation} 
The scheme is GDoF optimal if we let each entry of the channel matrix scale as $H_{ij} = \tilde{H}_{ij} P^{\alpha_{ij}}$ and 
\begin{equation}
\lim_{P\to\infty}\max_{\tilde{\pmb{r}}\in\mathcal{C}}\min_{\pmb{r}\in{\mathcal{R}}}\frac{\|\tilde{\pmb{r}}-\pmb{r}\|}{\log P} = 0,\ \forall\,\tilde{H}_{ij}, \alpha_{i,j}.
\end{equation}
{Note that any above optimality still holds when we scale the power $P$
by a constant. Therefore, throughout the paper, we scale the power
whenever it is convenient. }

\subsection{Linear precoding with point-to-point codes}
By removing the DPC from the transmitter, we have a much 
simpler but strictly suboptimal scheme, namely, the linear precoding
scheme. Specifically, by linear precoding, here we refer to a particular
class of schemes such that, 1) independent point-to-point
\emph{Gaussian} codebooks\footnote{A discussion on non-Gaussian
signaling is provided in Section~\ref{sec:LP}, Remark~\ref{remark:nonGaussian}.} are used to encode the $K$ streams; 2) the
transmitted signal is a linear combination of the $K$~codewords; and
3)~interferences are treated as noise at each receiver. 
Under these assumptions, a single-letter rate region can be obtained in
terms of the input random variable $\rvVec{X} = \sum_{k=1}^K \rvVec{X}_k$
with \emph{independent Gaussian distributed} $\{\rvVec{X}_k\}_k$ such
that $\E[\rvVec{X}_k \rvVec{X}_k^\H] = \Qm_k$ and $\sum_{k=1}^K
\trace(\Qm_k) \le P$. Then, the rate region achieved by such a linear precoding
is
\begin{IEEEeqnarray}{rCl}
  \IEEEeqnarraymulticol{3}{l}{
  \mathcal{C}_{\text{BC}}^{\text{LP}}(\{\Hm_{\!k}\}_k, P)} \nonumber \\
  \IEEEeqnarraymulticol{3}{l}{
  = \! \bigcup_{\{\Qm_k\}_k:\
  \sum_{k=1}^K \trace(\Qm_k)\le P} \biggl\{
  (R_1,\ldots,R_K) {\in \mathbb{R}_+^K}: } \nonumber \\
  R_k \le \log\det\biggl( \Id + \Bigl(\Id+\sum_{l\ne k}\Hm_{\!k}
  \Qm_l \Hm_{\!k}^\H\Bigr)^{-1} \Hm_{\!k} \Qm_k \Hm_{\!k}^\H \biggr)
  \biggr\}.\nonumber
\end{IEEEeqnarray}%

Note that the region $\mathcal{C}_{\text{BC}}^{\text{LP}}(\{\Hm_{\!k}\}_k,
P)$ is not convex. With a simple time-sharing strategy, we can achieve the 
convex hull of the region. The time-sharing strategy can also be
generalized to the resource-sharing strategy. Specifically, one can divide
the whole resource~(e.g., time and frequency) into orthogonal portions,
say, $\lambda_1,\ldots,\lambda_N$, such that
$\lambda_1+\cdots+\lambda_N=1$ and $\lambda_i>0$, $\forall\,i$. In each portion $i$ of the resource, we can
perform the linear precoding with covariance
matrices~$\{\Qm_{k}^{(i)}\}_k$. Instead of imposing that $\sum_{k=1}^K
\trace(\Qm_k^{(i)}) \le P$, $\forall\,i$, we only let $\sum_{i=1}^N
\lambda_i \sum_{k=1}^K \trace(\Qm_k^{(i)}) \le P$. Although the
resource-sharing strategy can improve the achievable rate region, we can show that
the improvement is bounded. 
\begin{lemma}\label{lemma:2}
  With linear precoding schemes, the achievable rate
  region with the resource-sharing strategy described above is within
  $\nr$~{bits/s/Hz} to the region with only time-sharing, that is, 
  \begin{align}
    \Conv \left\{ \mathcal{C}_{\text{BC}}^{\text{LP}}(\{\Hm_{\!k}\}_k,
  P) \right\}. 
  \end{align}%
\end{lemma}
\begin{proof}
  See Appendix~\ref{app:lemma2}. 
\end{proof}

Therefore, it is without loss of constant-gap optimality to focus on the simple
time-sharing strategy.

\subsection{Single-antenna~(SISO) BC}
\label{sec:SISO}

In the single-antenna case, i.e., when the transmitter and all the
receivers have each only one antenna, the analysis becomes easier.
We can prove that the rate region of the linear scheme 
$\Conv\left\{\mathcal{C}_{\text{BC}}^{\text{LP}}(\{h_k\}_k, P)\right\}$
is not constant-gap optimal. Let us consider the two-user case in which $|h_1|\gg|h_2|$.
From the MAC-BC duality, let $P_1=P_2=\frac{1}{2}P$, the following rates are
achievable
\begin{align}
  R_1 &= \log(1+P_1|h_1|^2 + P_2|h_2|^2) - \log(1+P_2|h_2|^2) \nonumber\\
   &\approx \log(1+P|h_1|^2) - \log(1+P|h_2|^2), \\
  R_2 &= \log(1+P_2|h_2|^2) \approx \log(1+P|h_2|^2),
\end{align}%
where we recall that ``$\approx$'' stands for constant-gap approximation. In contrast,
with the linear scheme, the achievable rate of user~2 is
$R'_2 = \log(1+\frac{\tilde{P}_2 |h_2|^2}{1+\tilde{P}_1 |h_2|^2})$ 
where $\tilde{P}_1, \tilde{P}_2$ are the power for user~1 and user~2,
respectively, with $\tilde{P}_1+\tilde{P}_2\le P$. 
In order for user~2 to achieve $\log(1+P|h_2|^2)$ {bits/s/Hz} with the linear scheme, 
the interference term $\tilde{P}_1|h_2|^2$ in $R'_2$ must remain bounded while the power
$\tilde{P}_2$ should be within a constant factor to $P$. Consequently, user~1's
rate must be bounded by a constant, since 
\begin{align}
  R'_1 &= \log(1+\frac{\tilde{P}_1 |h_1|^2}{1+\tilde{P}_2 |h_1|^2}) \\
   &\le \log(1+\frac{{P} |h_1|^2}{\tilde{P}_2 |h_1|^2})
  \\&\approx 0, 
\end{align}
where the last equality is from the fact that $\tilde{P}_2$ is within a
constant factor to $P$. 

Nevertheless, we know that single-user transmission achieves the sum capacity
to within a constant gap. Indeed, if we only serve the user with the
strongest channel gain, say, $|h_1| = \max_{k\in[K]} |h_k|$, then the sum rate
is 
\begin{align}
  \log(1+P|h_1|^2) &\ge 
  \log\Bigl(1+\frac{1}{K}P\sum_{k=1}^K |h_k|^2\Bigr) \\
  &\ge \log\Bigl(1+P\sum_{k=1}^K |h_k|^2\Bigr) - \log(K), 
\end{align}%
whereas the sum capacity is, from the MAC-BC duality, 
\begin{IEEEeqnarray}{rCl}
  \max_{\sum_k\!P_k\le P} \log\Bigl(1+\sum_{k=1}^K P_k |h_k|^2\Bigr) &\le  
  \log\Bigl(1+\sum_{k=1}^K P |h_k|^2\Bigr). \IEEEeqnarraynumspace 
\end{IEEEeqnarray}

\begin{remark}
Note that the capacity region of a SISO BC can be achieved by
linear superposition coding, since the channel is stochastically
degraded~\cite{NIT}. So a linear scheme does achieve capacity in this case.
However, the receivers need to decode a subset of the interfering
signals in order to achieve the capacity. Specifically,
each receiver needs to first decode all the messages for the receivers
with weaker channel gains, then remove the interference before decoding
the intended message. In fact, this is a simple form of RS
in which the message for the weakest user is indeed a common message
that needs to be decoded by~(although not intended to) all the 
users. The performance of linearly precoded RS in a general
MIMO setting is the main subject of this paper, and will be treated in
Section~\ref{sec:RS}. 
\end{remark}

\section{Rate-Splitting with MMSE Precoding}\label{sec:RS}

In this section, we introduce a RS scheme at the transmitter side, and
describe this scheme in the general MIMO case with $K$ users. We shall
derive the corresponding achievable rate region in its general form. 

\subsection{$K$-user BC with common messages} \label{sec:RS1}
The considered RS scheme builds on a general $K$-user scheme with
common messages. It is worth mentioning that the capacity region of the
two-user MIMO BC with common message has been completely characterized
in~\cite{geng2014capacity}. In that work, the authors showed that
Marton's inner bound based on binning is indeed tight with Gaussian
signaling. Here, we shall investigate the general $K$-user case but only
on the achievable rate region with independent point-to-point codebooks.   

First, let $\bigl\{\rv{M}_{\mathcal{K}}:\ \mathcal{K}\subseteq [K], \mathcal{K}\ne
\emptyset\bigr\}$ be a set of $2^K-1$ independent messages, each one with rate
$R_{\mathcal{K}}$ {bits/s/Hz}. 
These messages are encoded with independent Gaussian codebooks, each
generated identically and independently according to a
distribution~$\rv{X}_{\mathcal{K}}\sim\mathcal{CN}(0,\Qm_{\mathcal{K}})$,
$\forall\,\mathcal{K}\subseteq{[K]}, \mathcal{K}\ne \emptyset$, with
\begin{align}
  \Qm_{\mathcal{K}} &= \left( P_{\mathcal{K}}^{-1}\Id + \Hm_{\!\bar{\mathcal{K}}}^\H
  \Hm_{\!\bar{\mathcal{K}}} \right)^{-1},
\end{align}
 where $\bar{\mathcal{K}} := [K]\setminus \mathcal{K}$ and
$\Hm_{\!\bar{\mathcal{K}}}$ is a matrix formed by the vertical
concatenation of the channel matrices of the users in $\bar{\mathcal{K}}$,
with the convention $\Hm_{\emptyset} = 0$;
the coefficients $\{P_{\mathcal{K}}\}$ are chosen to satisfy the power
constraint $\sum_{\mathcal{K}} \trace\left( \Qm_{\mathcal{K}} \right)
\le P$.
Such a precoding scheme is known as the MMSE precoding. The idea behind
the MMSE precoding is to limit the interference power at the unintended
receivers. Indeed, the covariance matrix of
$\rvVec{X}_{\mathcal{K}}$ at the set $\bar{\mathcal{K}}$ of users is
{
\begin{IEEEeqnarray}{rCl}
  \E[\Hm_{\!\bar{\mathcal{K}}}\rvVec{X}_{\mathcal{K}}(\Hm_{\!\bar{\mathcal{K}}}\rvVec{X}_{\mathcal{K}})^\H] &=
  \Hm_{\!\bar{\mathcal{K}}}\left( P_{\mathcal{K}}^{-1}\Id + \Hm_{\!\bar{\mathcal{K}}}^\H
  \Hm_{\!\bar{\mathcal{K}}} \right)^{-1}\Hm_{\!\bar{\mathcal{K}}}^\H \preceq
  \Id,\IEEEeqnarraynumspace
\end{IEEEeqnarray}}%
that is, below the AWGN level. Unlike the ZF precoding that completely
nullifies interference, the MMSE precoding is known to achieve a better
trade-off between interference and signal power. Further, the application
of the ZF precoding is possible only when a non-empty interference null
space exists, whereas the MMSE precoding is feasible in general.
The transmitted signal is a superposition of all the streams
\begin{align}
\rvVec{X} = \sum_{\mathcal{K}\subseteq[K]}\rvVec{X}_{\mathcal{K}}. 
\end{align}

Next, each receiver $k$ jointly decodes the set of messages
$\{\rv{M}_{\mathcal{K}}:\,\mathcal{K}\ni k\}$ by treating the interferences
$\{\rvVec{X}_{\mathcal{K}'}:\, \mathcal{K}'\not\ni k\}$ as noise. Thus,
for each receiver $k$, it is
equivalent to a virtual MAC whose achievable rate region is the set of
{non-negative rate tuples} satisfying, for every collection $\collectionS_k \subseteq
\{\mathcal{K}:\, \mathcal{K}\ni k\}$,
\begin{IEEEeqnarray}{rCl}
  \IEEEeqnarraymulticol{3}{l}{\sum_{\mathcal{K}\in \collectionS_k}
  R_{\mathcal{K}}\le} \nonumber \\
  \log\det\biggl(\Id + \Bigl(\Id +\!\!\! \sum_{\mathcal{K}': \mathcal{K}'\not\ni k}
  \Hm_{\!k} \Qm_{\mathcal{K}'} \Hm_{\!k}^\H\Bigr)^{-1}
  \sum_{\mathcal{K}\in \collectionS_k} \Hm_{\!k} \Qm_{\mathcal{K}}
  \Hm_{\!k}^\H  \biggr).\IEEEeqnarraynumspace\label{eq:rate_con_exact}
\end{IEEEeqnarray}
The above rate constraints provide the exact characterization of the
achievable rate region for any linear precoding scheme~(not necessarily
the MMSE precoding). Note that the region is quite involved with a
large number of parameters. For our purpose, however, it is enough to
have an approximate region, i.e., to within a constant gap. This allows
us to simplify the region and obtain the following result.

\begin{lemma}\label{lemma:reduced_con}
  Let $\mathcal{{R}}_{\text{BC}}^{\text{CM}}(\{\Hm_{\!k}\}_k, P)$ be the set
  of achievable rate tuples $(R_{\mathcal{K}}:\, \mathcal{K}\subseteq[K],
  \mathcal{K}\ne\emptyset)$ by the proposed scheme with MMSE precoding satisfying the power
  constraint $P$ {and joint decoding at the receivers}. {Then, the
  set of non-negative rate tuples satisfying}
  \begin{align}
  \sum_{\mathcal{K}\in \collectionS_k} R_{\mathcal{K}} \le
  \log\det\biggl(\Id + \Hm_{\!k} \Qm_{\collectionS_k} \Hm_{\!k}^\H  \biggr),
  \label{eq:region0}
\end{align}
for all $k\in[K]$ and collections $\collectionS_k \subseteq \{\mathcal{K}:\,
\mathcal{K}\ni k\}$ {forms an achievable constant-gap rate region with
respect to $\mathcal{{R}}_{\text{BC}}^{\text{CM}}$,}
where we define for convenience
\begin{align}
  \Qm_{\collectionS_k} := 
  \sum_{\mathcal{K}\in \collectionS_k} \left(
   P^{-1}\Id + \Hm_{\!\bar{\mathcal{K}}}^\H \Hm_{\!\bar{\mathcal{K}}}
   \right)^{-1}.  \label{eq:def-Qs}
\end{align}
\end{lemma}

Note that in the above simplification, we have omitted the interference
term and replaced $P_{\mathcal{K}}$ by $P$ in $\Qm_{\mathcal{K}}$,
both of which only incur a bounded power loss in terms of $K$, but can
simplfy further analyses. The number of
constraints in the above region corresponds to the number of non-empty collections
$\collectionS_k$ for all $k\in[K]$.

We say that the collection $\collectionS_k$ is \emph{minimal} if no element is a proper subset of another
element. {For example, $\{ \{1\}, \{1,2\}, \{1,3\}\}$ as
$\collectionS_1$ is not minimal since
$\{1\}\subset \{1,2\}$ and $\{1\}\subset \{1,3\}$. One can always obtain a minimal collection by
removing the ``smaller'' elements, e.g., removing $\{1\}$ in the
previous example and we obtain $\{ \{1, 2\}, \{1,3\}\}$ that is minimal.} We
say that $\collectionS_k$ can be \emph{reduced} to a minimal collection
denoted by $\underline{\collectionS}_k$. It is readily shown that
{$|{\collectionS}_k| \Qm_{\underline{\collectionS}_k}  \succeq \Qm_{\collectionS_k} \succeq  \Qm_{\underline{\collectionS}_k}$} 
which we denote as
\begin{align}
\Qm_{\underline{\collectionS}_k} \!\! \approx \Qm_{{\collectionS}_k}.
\label{eq:minimal}
\end{align}
Therefore, we can replace $\Qm_{\collectionS_k}$ by
$\Qm_{\underline{\collectionS}_k}$ and only lose up to a constant number of
{bits per channel use}. Further, we notice that if both collections $\collectionS'_k$ and
$\collectionS''_k$ can be reduced to $\underline{\collectionS}_k$, then
$\collectionS'_k \bigcup \collectionS''_k$ can also be reduced to
$\underline{\collectionS}_k$. Hence, in the equivalent class of
collections sharing the same minimal $\underline{\collectionS}_k$, there is always a \emph{maximal}
collection that is the union of all the collections that can be reduced to
$\underline{\collectionS}_k$. It follows that for every collection
$\collectionS_k$, there is a minimal $\underline{\collectionS}_k$ and a
maximal $\overline{\collectionS}_k$ such that
\begin{align}
  \underline{\collectionS}_k \subseteq \collectionS_k \subseteq
  \overline{\collectionS}_k.
\end{align}%
For instance, when $K=2$, there are three possible collections for
$\collectionS_1$, namely, $\{\{1\}\}$, $\{\{1,\! 2\}\}$, and $\{\{1\},\!
\{1,\! 2\}\}$. Similar collections can be found for $\collectionS_2$. Note that $\{\{1,\! 2\}\}$
is both a $\collectionS_1$ and a $\collectionS_2$.
The three collections for $\collectionS_1$ can be divided into two
classes according to the minimal/maximal collection pairs as follows.
\renewcommand{\arraystretch}{1.2}
\begin{center}
  \begin{tabular}{|c|c|c|}
    \hline
    $\collectionS_1$ &  $\underline{\collectionS}_1$ &
    {$\overline{\collectionS}_1$} \\ \hline
$\{\{1\}\}$ & $\{\{1\}\}$ & $\{\{1\}\}$ \\ \hline
{$\{\{1,\! 2\}\}$} & \multirow{2}{*}{$\{\{1,\! 2\}\}$} &
\multirow{2}{*}{$\{\{1\},\! \{1,\! 2\}\}$} \\ \cline{1-1}
$\{\{1\},\! \{1,\! 2\}\}$ &  &  \\ \hline
  \end{tabular}
\end{center}
\renewcommand{\arraystretch}{1}
In the expression~\eqref{eq:region0}, we see that
among all the constraints with $\collectionS_k$ having the same
$\underline{\collectionS}_k$, thus having the same right hand side in
\eqref{eq:region0} up to a constant gap due to \eqref{eq:minimal},
the constraint corresponding to
$\overline{\collectionS}_k$ is obviously dominant since it involves all the
possible terms on the left hand side. Therefore, we can further simplify
the approximate rate region.
\begin{proposition} \label{pro:rate_region_common}
  The set $\CG{\mathcal{R}}_{\text{BC}}^{\text{CM}}(\{\Hm_{\!k}\}_k,
  P)$ of non-negative common message rate tuples satisfying, for all
  $k\in[K]$ and $\collectionS_k\subseteq\{\mathcal{K}:\,\mathcal{K}\ni
  k\}$, the rate constraints
  \begin{align}
  \sum_{\mathcal{K}\in \overline{\collectionS}_k} R_{\mathcal{K}} \le
  {l_k^{\underline{\collectionS}_k}} \label{eq:region1}
\end{align}
  is an achievable constant-gap rate region by the proposed scheme;
  here we define 
  \begin{align}
    l_k^{\collectionS} := \log\det\bigl(\Id + \Hm_{\!k} \Qm_{{\collectionS}} \Hm_{\!k}^\H
    \bigr). \label{eq:l_S} 
  \end{align}%
  
\end{proposition}

\begin{example}[The two-user case] \label{ex:2user}
When $K=2$, the rate region from
Proposition~\ref{pro:rate_region_common} becomes 
\begin{gather}
  \tilde{R}_1, \tilde{R}_{12}^{(1)}, \label{eq:ex1}
  \tilde{R}_2, \tilde{R}_{12}^{(2)} \ge 0, \\
  \tilde{R}_1 \le l_1^{\{1\}},  
  \tilde{R}_2 \le l_2^{\{2\}},\\ 
  \tilde{R}_1 + \tilde{R}_{12}^{(1)} + \tilde{R}_{12}^{(2)} \le
  l_1^{\{1,2\}},\\ 
  \tilde{R}_2 + \tilde{R}_{12}^{(1)} + \tilde{R}_{12}^{(2)} \le
  l_2^{\{1,2\}},  \label{eq:ex4} 
\end{gather}
where $l_k^{\{k\}} := \log\det\left( \Id + \Hm_k \Qm_{\{k\}} \Hm_k^\H
   \right)$ and  $l_k^{\{1,2\}} := \log\det\left( \Id + P \Hm_k \Hm_k^\H
   \right)$, for $k=1,2$; $\Qm_{\{1\}} := (P^{-1}\Id + \Hm_2^\H
   \Hm_2)^{-1}$ and $\Qm_{\{2\}} := (P^{-1}\Id + \Hm_1^\H \Hm_1)^{-1}$.
\end{example}

\subsection{$K$-user BC {without common} messages: Rate-splitting}

Now, let us get back to the original setting {without common}
messages, {i.e., with messages }
$\{\rv{M}_k:\,k=1,\ldots,K\}$, {each one intended exclusively to
one user. We can build a scheme without common messages from any scheme
with common messages through rate-splitting. }

First, we split each message $\rv{M}_k$ of rate $R_k$ {bits/s/Hz} into
sub-messages $\{\tilde{\rv{M}}^{(k)}_{\mathcal{K}}:\,\mathcal{K}\ni
k\}$, each of rate $\tilde{R}^{(k)}_{\mathcal{K}}$ such that
$\sum_{\mathcal{K}\ni k} \tilde{R}^{(k)}_{\mathcal{K}} = R_k$,
$\forall\,k$. By construction, each {sub-message}
$\tilde{\rv{M}}^{(k)}_{\mathcal{K}}$ should be decoded by all users in
$\mathcal{K}$, although the {sub-message} is intended only to user $k$.

Then, the $2^{K-1} K$ sub-messages $\bigl\{\tilde{\rv{M}}^{(k)}_{\mathcal{K}}:\
k\in\mathcal{K}\subseteq[K]\bigr\}$ are re-assembled into $2^K-1$
{sub-messages}
\begin{align}
  \tilde{\rv{M}}_{\mathcal{K}} := \{
  \tilde{\rv{M}}^{(k)}_{\mathcal{K}}:\,k\in\mathcal{K}\},
\end{align}
each one of which should be decodable by the users in $\mathcal{K}$ by
construction. These $2^K-1$ re-assembled {sub-messages} are transmitted with
the scheme described in the previous subsection. At the receivers' side,
each user $k$ decodes the set of re-assembled {sub-messages}
$\{\tilde{\rv{M}}_{\mathcal{K}}:\, \mathcal{K}\ni k\}$, but only keeps
the sub-messages $\{\tilde{\rv{M}}^{(k)}_{\mathcal{K}}:\,\mathcal{K}\ni
k\}$ in order to reconstruct the desirable message $\rv{M}_k$.
At this point, the following result becomes straightforward.
\begin{proposition}\label{pro:rate_region}
  A {rate tuple} $(R_1,\ldots,R_K)$ is achievable if there
  exists a set of {sub-message} rates $\bigl\{\tilde{R}^{(k)}_{\mathcal{K}}\!:
  k\in\mathcal{K}\subseteq[K]\bigr\}$ such that
  \begin{gather}
    \sum_{\mathcal{K}\ni k} \tilde{R}^{(k)}_{\mathcal{K}} = R_k, \
    \forall k\in[K],\\
    \sum_{k \in \mathcal{K}} \tilde{R}^{(k)}_{\mathcal{K}} =
    \tilde{R}_{\mathcal{K}}, \ \forall \mathcal{K}\subseteq[K], \\
    \text{and } \bigl(\tilde{R}_{\mathcal{K}}:\,
    \mathcal{K}\subseteq[K]\bigr) \in
    \mathcal{R}_{\text{BC}}^{\text{CM}}(\{\Hm_{\!k}\}_k, P).
    \label{eq:tmp321}
  \end{gather}
The set of such rate tuples is denoted by
$\mathcal{R}_{\text{BC}}^{\text{RS}}(\{\Hm_{\!k}\}_k,P)$. {Replacing the
rate region $\mathcal{R}_{\text{BC}}^{\text{CM}}$ in \eqref{eq:tmp321} by the constant-gap
rate region $\CG{\mathcal{R}}_{\text{BC}}^{\text{CM}}$, we obtain an
achievable constant-gap rate region $\CG{\mathcal{R}}_{\text{BC}}^{\text{RS}}(\{\Hm_{\!k}\}_k,P)$ with respect to
$\mathcal{R}_{\text{BC}}^{\text{RS}}(\{\Hm_{\!k}\}_k,P)$.} 
\end{proposition}

It is worth emphasizing the three choices that we have made for the above RS
scheme: 1)~independent codebooks for different sets of
{sub-messages},
2)~linear spatial MMSE precoding at the transmitter, and 3)~decoding
common interfering streams by treating other streams as noise at the
receivers.
Since the proposed scheme allows the receivers to decode partially the
interference, it can be regarded as a receiver-side interference
mitigation scheme.

Finally, it is possible to ignore a subset
$\mathcal{T}\subseteq[K]$ of users and only apply the proposed
RS scheme to the remaining users in $[K]\setminus\mathcal{T}$. Together
with time sharing, the achievable rate region is described as follows.
\begin{corollary}\label{cor:TS}
  The following convex hull of rate-tuples is achievable with the proposed RS scheme and time sharing
  \begin{IEEEeqnarray}{rCl}
    \Conv\bigcup_{\mathcal{T}\subseteq[K]}\!\!\!\left\{
    R_{\mathcal{T}} = 0,
     R_{[K]\setminus\mathcal{T}} \in
     \mathcal{{R}}_{\text{BC}}^{\text{RS}}(\{\Hm_{\!k}\}_{k\in[K]\setminus\mathcal{T}},P)
    \right\}.\IEEEeqnarraynumspace
  \end{IEEEeqnarray}
  {Similarly, replacing $\mathcal{{R}}_{\text{BC}}^{\text{RS}}$ by
  $\CG{\mathcal{{R}}}_{\text{BC}}^{\text{RS}}$, we obtain an achievable
  constant-gap region.}
\end{corollary}

In the following, we shall only focus on constant-gap rates for our
purpose. For brevity, we drop the term ``constant-gap'' whenever
confusion is not likely.

\section{Achievable Sum Rate} \label{sec:sum-rate}

In general, the achievable rate region is analytically intractable. 
Even the numerical evaluation becomes hard for a moderately large number
of users due to the exponential growth of the number of
sub-messages and doubly exponential growth of the number of
constraints. In this section, we shall focus on the achievable
sum rate instead of the entire rate region to obtain meaningful insights. 
First, we shall establish an upper bound on the achievable sum rate
of the proposed region. We shall then show how to preserve the achievable sum rate up to a constant loss while reducing the total number of active streams.

\subsection{Sum rate upper bound}\label{sec:Kuser}

Before analyzing the sum rate, we take a closer look at the
defining term $l_k^{\collectionS}:= \log\det\bigl(\Id + \Hm_{\!k} \Qm_{{\collectionS}} \Hm_{\!k}^\H
    \bigr)$ of the achievable rate
region in Proposition~\ref{pro:rate_region_common}. 
First, note that $l_k^{\collectionS}$ is increasing with respect to the partial
ordering of the collection $\collectionS$. Indeed, when
$\collectionS\supseteq \collectionS'$, we have $\Qm_{\collectionS}
\succeq \Qm_{\collectionS'}$ according to the
definition~\eqref{eq:def-Qs}, implying $l_k^{\collectionS} \ge
l_k^{\collectionS'}$. Similarly, for $\collectionS =
\{\mathcal{K}\}$ and $\collectionS' = \{\mathcal{K}'\}$ with $\mathcal{K}\supseteq
\mathcal{K}'$, we have $l_k^{\collectionS} \ge l_k^{\collectionS'}$. 
For convenience, we also define 
\begin{align}
  C_{\mathcal{K}} &:=  \log\det\left( \Id + P \Hm_{\mathcal{K}}
  \Hm_{\mathcal{K}}^\H \right), \quad \forall\,\mathcal{K}\subseteq [K],
  \label{eq:def-C}
\end{align}
where we let $C_{\emptyset} := 0$. Finally, the following relationship between
the $C$'s and the $l$'s will be useful.  
\begin{lemma}\label{lemma:l-C}
  For each $k\in[K]$ and each collection $\collectionS$ of subsets of
  $[K]$, we have 
  \begin{align}
    l_k^{\collectionS} &\ge \max_{\mathcal{K} \in \collectionS}
    C_{\{k\}\cup\bar{\mathcal{K}}} - C_{\bar{\mathcal{K}}}, 
  \end{align}%
  where equality holds when $|\collectionS|=1$. 
\end{lemma}
\begin{proof}
  Since from the definition \eqref{eq:def-Qs} $\Qm_{\collectionS} := \sum_{\mathcal{K}\in\collectionS} \bigl(P^{-1}\Id +
  \Hm_{\bar{\mathcal{K}}}^\H \Hm_{\bar{\mathcal{K}}} \bigr)^{-1} \succeq \bigl(P^{-1}\Id +
  \Hm_{\bar{\mathcal{K}}}^\H \Hm_{\bar{\mathcal{K}}} \bigr)^{-1}$,
  $\forall$ $\mathcal{K}\in\collectionS$, we have   
$l_k^{\collectionS} := \log\det\left( \Id + \Hm_k \Qm_{\collectionS}
  \Hm_k^\H \right) \ge \log\det\left( \Id + \Hm_k \bigl(P^{-1}\Id +
  \Hm_{\bar{\mathcal{K}}}^\H \Hm_{\bar{\mathcal{K}}} \bigr)^{-1} 
  \Hm_k^\H \right) =  C_{\{k\}\cup\bar{\mathcal{K}}} -
  C_{\bar{\mathcal{K}}}$. Equality holds when $|\collectionS|=1$. 
\end{proof}

\begin{proposition}\label{prop:sumrate}
{The maximum sum rate of the RS scheme with $K$ active users is upper bounded by}
\begin{multline}
  { \min }\Biggl\{  {C_{[K]},}  \sum_{i=1}^K \left(
  \frac{l_i^{(K)}}{K-1} + \sum_{k=1}^{K-2} \frac{l_i^{(k)}}{k(k+1)}
  \right)  \\
  + \frac{1}{K-1}\min_{m\in[K]} \left\{ l_m^{(K-1)} - l_m^{(K)}
  \right\} \Biggr\},
  \label{eq:UB-Kuser}
\end{multline}
where 
\begin{align}
  l_i^{(k)} &:= {l_i^{\underline{\collectionS}_i^{(k)}}
  ={} } \log\det\left( \Id + \Hm_{\!i}
  \Qm_{\underline{\collectionS}_i^{(k)}} \Hm_{\!i}^\H \right),
  \label{eq:def_l} \\
  \underline{\collectionS}_i^{(k)} &:= \left\{ \mathcal{S}:\
  |\mathcal{S}|=k,\,i\in\mathcal{S} \right\},\quad
  \forall\,i\in[K], \label{eq:def-Sik} 
\end{align}
{where we recall that $\Qm_{\collectionS}$ is defined as in
\eqref{eq:def-Qs}. In particular, the upper bound is achievable to
within a constant gap when $K\le3$.}
\end{proposition}
\begin{proof}
  {The first upper bound $C_{[K]}$ is trivial since it is the
  sum capacity of the channel. Alternatively, we can also recover it
  from the region~\eqref{eq:region1}. Indeed, let us consider the
  sequence of maximal collections
  $\overline{\collectionS}_k=\{\mathcal{S}: k\in\mathcal{S},\,
  i\not\in\mathcal{S},\, \forall i>k\}$, $k\in[K]$, and their
  correponding minimal collections $\underline{\collectionS}_k=\{[k]\}$,
  $k\in[K]$. Then, the sum rate can be decomposed as
  $\sum_{k\in[K]} \sum_{\mathcal{K}\in\overline{\collectionS}_k}
  \tilde{R}_\mathcal{K} \le \sum_{k\in[K]}
  l_k^{\underline{\collectionS}_k} = \sum_{k\in[K]} C_{\{k,\ldots,K\}} -
  C_{\{k+1,\ldots,K\}} = C_{[K]}$ where the inequality is from
  \eqref{eq:region1} and the first equality is from
  Lemma~\ref{lemma:l-C}. We shall now focus on the second term in
  \eqref{eq:UB-Kuser}.
  }

  For given $i,k\in[K]$, let us consider the collection
  $\underline{\collectionS}_i^{(k)}$ as defined in \eqref{eq:def-Sik}. It
  is clearly a minimal collection according to the definition in
  Section~\ref{sec:RS1}, since no element is a proper subset of another
  element. Now, let us define the following collection
  \begin{equation}
  \overline{\collectionS}_i^{(k)} := \left\{ \mathcal{S}:\
  1\le|\mathcal{S}|\le k,\,i\in\mathcal{S} \right\}.
  \label{eq:def-Sik-max} 
  \end{equation}
  We notice that $\overline{\collectionS}_i^{(k)} \supseteq
  \underline{\collectionS}_i^{(k)}$ is a maximal collection according to the definition in
  Section~\ref{sec:RS1}. From \eqref{eq:region1}, we have 
  \begin{equation}
    \sum_{\mathcal{K}\in \overline{\collectionS}_i^{(k)}} R_{\mathcal{K}}
    \le \log\det\biggl(\Id + \Hm_{\!k}
    \Qm_{\underline{\collectionS}_i^{(k)}} \Hm_{\!k}^\H \biggr) =
    l_i^{(k)}. \label{eq:tmp121}
  \end{equation}
  Letting $k=K$ in \eqref{eq:tmp121}, we obtain
  \begin{equation}
    \underbrace{
    R_{[K]} + \sum_{k'=1}^{K-1}
    \sum_{\mathcal{K}\in\underline{\collectionS}_i^{(k')}}
    R_{\mathcal{K}}}_{a_i} \le l_i^{(K)}, \label{eq:ai}
  \end{equation}
  while letting $k=K-1$ in \eqref{eq:tmp121}, we have
  \begin{equation}
    \underbrace{
    \sum_{k'=1}^{K-1} \sum_{\mathcal{K}\in\underline{\collectionS}_i^{(k')}}
    R_{\mathcal{K}}}_{b_i} \le l_i^{(K-1)}.\label{eq:bi} 
  \end{equation}
  Next, we sum up $a_i$ and $b_i$ as follows.
  \begin{align}
    \sum_{i=1}^{K-1} a_i + b_K &= (K-1)R_{[K]} + \sum_{i=1}^K
    \sum_{k'=1}^{K-1}\sum_{\mathcal{K}\in
    \underline{\collectionS}_i^{(k')}} R_{\mathcal{K}} \label{eq:tmp111} \\
    &= (K-1)R_{[K]} + \sum_{k'=1}^{K-1} \sum_{i=1}^K
    \sum_{\mathcal{K}\in \underline{\collectionS}_i^{(k')}} R_{\mathcal{K}} \\
    &= (K-1)R_{[K]} + \sum_{k'=1}^{K-1} k' \sum_{\mathcal{K}:\,|\mathcal{K}|=k'} R_{\mathcal{K}} \\
    &= (K-1)R^{(K)} + \sum_{k'=1}^{K-1} k' R^{(k')},
  \end{align}
  where in the second equality we exchange the sums over $i$ and
  $k'$; in the third one we apply the symmetry and rearrange the
  $K\times\binom{K-1}{k'-1}$ summands into $k'$ summations over
  $\binom{K}{k'}$ terms; in the last one we define $R^{(k)} :=
  \sum_{\mathcal{K}:\,|\mathcal{K}|=k} R_{\mathcal{K}}$. Summing up the
  right-hand sides of \eqref{eq:ai} and \eqref{eq:bi} in the same way,
  we obtain
  \begin{align}
    (K-1)R^{(K)} + \sum_{k'=1}^{K-1} k' R^{(k')} &\le \sum_{i=1}^{K-1}
    l_i^{(K)} + l_K^{(K-1)}. \label{eq:tmp732} 
  \end{align}
  Similarly, for any $k\le K-2$, we can apply \eqref{eq:tmp121} and
  obtain
  \begin{equation}
    \sum_{k'=1}^{k} \sum_{\mathcal{K}\in\underline{\collectionS}_i^{(k')}} R_{\mathcal{K}} \le l_i^{(k)}.\label{eq:ci} 
  \end{equation}
  Summing over $i$, we have 
  \begin{equation}
    \sum_{i=1}^K \sum_{k'=1}^{k}
    \sum_{\mathcal{K}\in\underline{\collectionS}_i^{(k')}}
    R_{\mathcal{K}} = \sum_{k'=1}^{K-1} k' R^{(k')}.
  \end{equation}
  Since $\sum_{i=1}^K \sum_{k'=1}^{k}
    \sum_{\mathcal{K}\in\underline{\collectionS}_i^{(k')}}
    R_{\mathcal{K}} = \sum_{k'=1}^{K-1} k' R^{(k')}$,  we have
  \begin{equation}
    \underbrace{\sum_{k'=1}^{k} k' R^{(k')}}_{c_k} \le \sum_{i=1}^{K} l_i^{(k)}, \quad
    k=1,\ldots,K-2. \label{eq:tmp642}
  \end{equation}
  Consider the following weighted sum over $c_k$
  \begin{align}
    \sum_{k=1}^{K-2} \frac{K-1}{k(k+1)} c_k &= \sum_{k=1}^{K-2}
    \frac{K-1}{k(k+1)} \sum_{k'=1}^{k} k' R^{(k')} \\
    &= \sum_{k'=1}^{K-2} k' R^{(k')} \sum_{k=k}^{K-2} \frac{K-1}{k(k+1)} \\
    &= \sum_{k'=1}^{K-2} (K-1-k') R^{(k')},
  \end{align}
  and the same weighted sum over the right-hand side of
  \eqref{eq:tmp642}, we obtain
  \begin{align}
   \sum_{k'=1}^{K-2} (K-1-k') R^{(k')} \le \sum_{k=1}^{K-2}
   \sum_{i=1}^{K} \frac{K-1}{k(k+1)} l_i^{(k)}. \label{eq:tmp901}
  \end{align}
  Summing up \eqref{eq:tmp732} and \eqref{eq:tmp901}, and dividing both
  sides by $K-1$, we have 
  \begin{IEEEeqnarray}{rCl}
    {R}^{(K)}_{\text{sum}} &\le& \sum_{i=1}^K \left(
    \frac{l_i^{(K)}}{K-1} + \sum_{k=1}^{K-2} \frac{l_i^{(k)}}{k(k+1)}
    \right)\nonumber \\ &&+\> \frac{1}{K-1} \left( l_K^{(K-1)} - l_K^{(K)} \right),
  \end{IEEEeqnarray}
  where ${R}^{(K)}_{\text{sum}} := \sum_{i=1}^K R^{(i)}$. 
  Note that in \eqref{eq:tmp111}, we can also start with $\sum_{i=1,i\ne
  m}^K a_i + b_m$
  instead of $\sum_{i=1}^{K-1} a_i + b_K$, for any $m=1,\ldots,K$.
  Repeating the same steps, we can obtain 
  \begin{IEEEeqnarray}{rCl}
    {R}^{(K)}_{\text{sum}} &\le& \sum_{i=1}^K \left(
    \frac{l_i^{(K)}}{K-1} + \sum_{k=1}^{K-2} \frac{l_i^{(k)}}{k(k+1)}
    \right)\nonumber\\
    &&+\> \frac{1}{K-1} \left( l_m^{(K-1)} - l_m^{(K)} \right),
  \quad m\in[K]. 
  \end{IEEEeqnarray}
  {Since the above upper bound holds for every rate tuple, it is also
  an upper bound on the maximum sum rate. This concludes the proof of
  the upper bound \eqref{eq:UB-Kuser}. We defer the proof of the
  achievability for $K\le3$ in Appendix~\ref{app:prop5}.} 
\end{proof}
In Table~\ref{tab:gap}, we provide the numerical evaluation of the upper
bounds and the sum rate of the region~\eqref{eq:region1} for $K=4$ and
$5$ users.\footnote{With more users, numerical simulations become
infeasible due to a large amount of constraints, up to the order
$2^{2^{K-1}}$.} Two channel models are considered: Rayleigh fading and
the one-ring model~\cite{adhikary2013joint}.  Average rates are obtained
with $1000$ channel realizations for each distribution, while maximum
gap is from all $2000$ realizations.  We observe that the upper bound is
indeed very close to the exact rate in average, and the maximum gap is
small as compared to the average rate. 
\begin{table*}[t]
  \caption{{Numerical results comparing the derived upper
  bound~\eqref{eq:UB-Kuser} and the
  exact RS sum rate of the region~\eqref{eq:region1}~(in bits/s/Hz) for
  $P=40~\mathrm{dB}$. The transmitter has $\nt=6$ antennas. } }
  {
\begin{center}
\begin{tabular}{|c|c|c|c|c|c|}
    \hline
    \multirow{2}{*}{}&\multicolumn{2}{c|}{Rayleigh, average}
    &\multicolumn{2}{c|}{One ring, average} & \multirow{2}{*}{maximum
    gap} \\ \cline{2-5}
    &upper bound & simulation &upper bound & simulation &  \\ \hline
    $K=4$ & 64.7517&64.7404 &31.4040 &31.0629 & 0.7236\\ \hline
    $K=5$ & 80.5573&80.5372 &33.5408 &33.1612 &1.8441 \\ \hline
\end{tabular}
\end{center}
\label{tab:gap}}
\end{table*}

\begin{remark}
  Note that the upper bound in \eqref{eq:UB-Kuser} holds when all $K$
  users are active. Nevertheless, one can ignore $K-K'$ users and apply the
RS scheme to the $K'$ active users for any $K'\le K$. In this case, the
above bound is still valid by replacing $K$ with $K'$ and replacing
$\Hm$ with the corresponding submatrix. 
\end{remark}

\subsection{Stream elimination and stream ordering algorithms}\label{sec:Algo}

The general RS scheme transforms $K$ {messages} into
$2^K-1$ different {sub-messages}, and then the BS creates one
stream for each {sub-message} aiming at the corresponding user
group. In practice, we would like to reduce the number of streams for 
lower signaling and decoding complexity. 
In this section, we first propose an algorithm that eliminates some of
the streams without reducing the sum rate for more than a given number
of bits {per channel use}. Then, based on the same idea, we propose a second algorithm
that orders all the $2^K-1$ streams, and validate the algorithm through
numerical simulation. For simplicity of demonstration,
we focus on the MISO case with $M\ge K$ in the following.  Nevertheless,
the results can {be} extended to multi-antenna receivers straightforwardly. 

\subsubsection{Sufficient conditions to maintain the sum rate}

For convenience, let us introduce some notations first. 
Let $\Hm^\dag$ be the Moore-Penrose inverse\footnote{Namely, $\Hm^\dag$ is such that 
$\Hm^\dag\Hm\Hm^\dag = \Hm^\dag$, $\Hm\Hm^\dag\Hm = \Hm$, and both
$\Hm^\dag\Hm$ and
$\Hm\Hm^\dag$ are Hermitian.} of $\Hm$. 
For any $\Ic\subseteq [K]$, let $\Hm_{\Ic} \in \mathbb{C}^{i\times M}$, with $i=|\Ic|$, denote 
the submatrix of $\Hm$ containing the rows with indices in
$\Ic$; similarly, let $\Hm^\dag_{\Ic} \in \mathbb{C}^{M\times i}$ denote the
submatrix of $\Hm^\dag$ containing the columns with indices in $\Ic$.

\begin{proposition}\label{prop:elim}
  For any $\Ic\subset [K]$ with $1\le|\Ic|<K$, if there exists some $\Wm
  \in\mathbb{C}^{M\times M}$ with $\|\Wm\|^2 \le c$ such that
  \begin{equation}
    \Hm_{\!\Ic} \Wm = \Hm_{\!\Ic}, \quad \text{and}\quad
    \Hm_{\!\bar{\Ic}} \Wm = \zerov, \label{eq:suff_cond}
  \end{equation}
  then we can eliminate any stream $\mathcal{K}$ such that
  $\mathcal{K}\cap \Ic \ne \emptyset$ and 
  $\mathcal{K}\cap \bar{\Ic} \ne \emptyset$ 
  without losing more than $\log(c)$ bits. 
\end{proposition}
\begin{proof}
  Let $i=|\Ic| \in [1:K-1]$ and let us assume that $\Ic =
  \{1,\ldots,i\}$ without loss of generality. The case with an arbitrary
  $\Ic$ follows in the same way up to a row permutation of $\Hm$.  Let
  $\xv_{\mathcal{K}}$ be the signal corresponding to the stream intended
  for user group $\mathcal{K}$, and the corresponding
  {sub-message} and rate
  are denoted by $M_{\mathcal{K}}$ and $R_{\mathcal{K}}$, respectively.
  Further, let us assume that $\mathcal{K} \cap
  \Ic \ne \emptyset$ and $\mathcal{K} \cap \bar{\Ic} \ne \emptyset$,
  i.e., $\mathcal{K}$ contains at least one user
  inside $\Ic$ and at least one outside of it. Without loss
  of generality, let $\mathcal{K} = \mathcal{K}'\cup\mathcal{K}''$ with
  $\mathcal{K}'\subseteq \Ic$ and $\mathcal{K}''\subseteq \bar{\Ic}$. 

  Let $\Wm \in\mathbb{C}^{M\times M}$ be such that  
  \eqref{eq:suff_cond} holds. Then, we 
  define a new signal corresponding to the {sub-message}
  $M_{\mathcal{K}}$ as $\xv_{\mathcal{K}}' = \Wm \xv_{\mathcal{K}}$. 
  Due to the condition \eqref{eq:suff_cond}, the received signal
  corresponding to the {sub-message} $M_{\mathcal{K}}$ at users
  in $\bar{\Ic}$ is $\Hm_{\!\bar{\Ic}} \xv_{\mathcal{K}}' =
  \Hm_{\!\bar{\Ic}} \Wm \xv_{\mathcal{K}} = \zerov$,
  while the received signal at users in $\Ic$ is $\Hm_{\!{\Ic}} \xv_{\mathcal{K}}' =
  \Hm_{\!{\Ic}} \Wm \xv_{\mathcal{K}} = \Hm_{\!{\Ic}}
  \xv_{\mathcal{K}}$, i.e., remains the same as with
  $\xv_{\mathcal{K}}$.  Hence, with the new signaling scheme, users in
  $\Ic$ see no changes, and  
  users in
  $\bar{\mathcal{I}}$ do not receive any signal related to the
  {sub-message}
  $M_{\mathcal{K}}$. 
  In other words, the {sub-message} $M_{\mathcal{K}}$ can be
  downgraded to a {sub-message} to
  users in $\mathcal{K}' = \Kc\cap\Ic$ without degrading the decoding performance of other
  users. 

  Next, we evaluate the power loss. 
  \begin{IEEEeqnarray}{rCl}
    \|\xv_{\mathcal{K}}'\|^2 &=  \| \Wm {\xv}_{\Kc} \|^2 \le
    \sigma_{\max}^2(\Wm) \| \xv_{\Kc} \|^2
    \le \|\Wm\|^2  \|\xv_{\Kc}\|^2, \IEEEeqnarraynumspace
  \end{IEEEeqnarray}
  where $\sigma_{\max}$ denotes the maximum singular value of a matrix, 
  with $\sigma_{\max}(\Wm)\le \|\Wm\|$. 
  Next we scale down the power of $\xv'_{\mathcal{K}}$ to meet the power constraint, namely, we let $\xv''_{\mathcal{K}} :=
  \frac{1}{\|\Wm\|} \xv'_{\mathcal{K}}$. Note that scaling down the
  power by a factor $\|\Wm\|^2$, we have a rate loss on
  $R_{\mathcal{K}}$ of at most $\log(\|\Wm\|^2)$ {bits/s/Hz}. 
  Since decreasing the power of one stream cannot hurt the other streams, the sum rate
  loss is at most $\log(\|\Wm\|^2)\le \log(c)$ {bits/s/Hz}. The proof is complete.   
\end{proof}

\subsubsection{Stream elimination and ordering}

\begin{lemma}
If the linear system \eqref{eq:suff_cond} has at least one solution, then the solution  
$\Wm_{\!\Ic} := \Hm_{\Ic}^\dag \Hm_{\Ic}$ is the one with minimum Euclidean norm. 
\end{lemma}
\begin{proof}
  Let $\tilde{\Hm} := \left[ \begin{smallmatrix}
    \Hm_{\!\Ic}\\\Hm_{\!\bar{\Ic}}\end{smallmatrix}\right] = \pmb{\Pi}
    \Hm$ for some permutation matrix $\pmb{\Pi}$. Then, the
    condition~\eqref{eq:suff_cond} can be rewritten as  
$\tilde{\Hm} \Wm  = \left[ \begin{smallmatrix}
    \Hm_{\!\Ic}\\\zerov\end{smallmatrix}\right]$. If this equation has at least one solution,
    then it is known~\cite{HJ:85} that $\tilde{\Hm}^\dag \left[ \begin{smallmatrix}
    \Hm_{\!\Ic}\\\zerov\end{smallmatrix}\right]$ has the minimum
    Euclidean norm. Since $\tilde{\Hm}^\dag =
    {\Hm}^\dag \pmb{\Pi}^\T$, we have $\Wm = {\Hm}^\dag \pmb{\Pi}^\T \left[ \begin{smallmatrix}
      \Hm_{\!\Ic}\\\zerov\end{smallmatrix}\right] = \Hm_{\Ic}^\dag
      \Hm_{\Ic}$. 
\end{proof}

To include the case where \eqref{eq:suff_cond}
does not have any solution, let us define  
\begin{equation}
  \widetilde{\Wm}_{\!\Ic} := \begin{cases} \Hm_{\Ic}^\dag \Hm_{\Ic}, & \text{when \eqref{eq:suff_cond} has a solution}, \\
    \pmb{\infty}, & \text{otherwise},
  \end{cases} 
  \label{eq:MP}
\end{equation}
where, with a slight abuse of notation, we use $\pmb{\infty}$ to denote a
matrix with infinite norm. 
Proposition~\ref{prop:elim} has the following equivalent form. 
\begin{corollary}\label{cor:elim}
  If $\bigl\|\widetilde{\Wm}_{\!\Ic} \bigr \|^2 \le c$
  for some non-empty set $\mathcal{I}\subset[K]$ and some non-negative
  value $c$, 
  then we can eliminate any stream $\mathcal{K}$ such that 
  $\mathcal{K}\cap \mathcal{I} \ne \emptyset$ and 
  $\mathcal{K}\cap \bar{\mathcal{I}} \ne \emptyset$
 without losing more than $\log(c)$ bits. 
\end{corollary}

At this point, we can describe our stream elimination algorithm in
Algorithm~\ref{algo:SEA}. 
\begin{algorithm}[t]
\caption{Stream elimination algorithm.}
\label{algo:SEA}
\begin{algorithmic}
  \State \underline{Input}: channel matrix $\Hm$, threshold $c$
  \State Initialize the collections: $\pmb{\mathcal{S}}(c) = 2^{[K]}\setminus
  \emptyset$
  \State Compute $\Hm^\dag$
  \For{non-empty subset $\mathcal{I} \subset [K]$ }
  \State Compute ${\Wm}_{\!\Ic} = \Hm_{\Ic}^\dag \Hm_{\Ic}$ 
  \If{$\Wm_{\!\Ic}$ verifies \eqref{eq:suff_cond}, and
  $\|{\Wm}_{\!\mathcal{I}}\|^2\le c$}
  \State Remove all $\mathcal{K}$ such that $\mathcal{K}\cap \mathcal{I} \ne \emptyset$ and 
  $\mathcal{K}\cap \bar{\mathcal{I}} \ne \emptyset$ from the
  collection $\pmb{\mathcal{S}}(c)$ 
  \EndIf
  \EndFor
  \State \underline{Output}: $\pmb{\mathcal{S}}(c)$
\end{algorithmic}
\end{algorithm}
Note that for each of the $2^K-2$ non-empty proper
subsets $\mathcal{I}$ of $[K]$, the complexity of finding $\Wm_{\!\Ic}$ is
$O(K^2M)$, while the complexity to verify \eqref{eq:suff_cond} is
$O(K M^2)$. Therefore, the overall complexity of Algorithm~\ref{algo:SEA} is
$O(K M^2 2^K)$, since $M\ge K$.

The following property is straightforward from the algorithm. 
\begin{claim}
  The output collection from Algorithm~\ref{algo:SEA}, denoted as
  $\pmb{\mathcal{S}}(c)$, is decreasing with the
  threshold $c$ such that $\pmb{\mathcal{S}}(0) = 2^{[K]}\setminus \emptyset$
  and $\pmb{\mathcal{S}}(\infty) = [K]$. 
\end{claim}

In practice, in order to reduce the precoding and decoding complexity, 
one may want to order the streams somehow and use only the
``best'' ones. From the
above discussion, we observe that one way to order the streams is to use
the minimum threshold for a stream to be eliminated from
Algorithm~\ref{algo:SEA}. Specifically, such a threshold is defined as,
for each $\Kc\subseteq[K]$ with $|\Kc|\ge2$, 
\begin{align}
  c_{\Kc} :\!&= \min_{\Ic:\,\Ic\cap\Kc\ne\emptyset,
  \bar{\Ic}\cap\Kc\ne\emptyset} \| \widetilde{\Wm}_{\!\Ic} \|^2
  \label{eq:cK0} \\
  &= \min_{\Ic_1\subset\Kc, \Ic_1\ne\emptyset
  \atop
  \Ic_2\subseteq\bar{\Kc}} \| \widetilde{\Wm}_{\!\Ic_1\cup\Ic_2}
  \|^2. \label{eq:cK} 
\end{align}
It is straightforward to verify that setting $c=c_{\Kc}$ can eliminate
the stream $\Kc$ with Algorithm~\ref{algo:SEA}. Then, one can order the
common streams $\Kc\subseteq[K]$, $|\Kc|\ge2$, according to the values
$\{c_{\Kc}\}_{\Kc\subseteq[K], |\Kc|\ge2}$. Note that there are in total
$2^K-1-K$ common streams and, hence, $2^K-1-K$ variables $c_{\Kc}$.
Since they can only take one of the $2^{K}-2$ values in
$\{\|\widetilde{\Wm}_{\!\Ic}\|^2\}_{\Ic\subset[K], \Ic\neq\emptyset}$, it is probable
that more than one common stream share the same threshold value. In that
case, one can introduce a simple randomization to resolve the tie
situation. 
Algorithm~\ref{algo:SOA} summarizes this procedure. 

\begin{algorithm}[t]
\caption{Stream ordering algorithm.}
\label{algo:SOA}
\begin{algorithmic}
  \State \underline{Input}: channel matrix $\Hm$, randomization
  parameter $\sigma$
  \State Compute $\Hm^\dag$
  \For{non-empty subset $\mathcal{I} \subset [K]$ }
  \State Compute $\|\widetilde{\Wm}_\mathcal{I}\|^2$ according to
  \eqref{eq:MP}
  \EndFor
  \For{non-empty subset $\mathcal{K} \subseteq [K]$ with $|\Kc|\ge2$ }
  \State Compute $c_{\Kc}$ according to \eqref{eq:cK} 
  \State Randomization: let $c_{\Kc}' = c_{\Kc} + \rv{D}_{\Kc}$ where
  $\rv{D}_{\Kc}$ is a zero-mean uniform random variable with variance $\sigma^2$ 
  \EndFor
  \State Sort $\{c'_{\Kc}\}_{\Kc\subseteq[K]}$ in decreasing order
  \State \underline{Output}: ordered streams
\end{algorithmic}
\end{algorithm}

Note that streams with larger threshold values are considered
``better''. The following claim is straightforward from the definition of the
minimum threshold in \eqref{eq:cK0}. 
\begin{claim}
  The minimum threshold $c_{\Kc}$, for $\Kc\subseteq[K]$ and
  $|\Kc|\ge2$, as defined in \eqref{eq:cK0}, is decreasing with the partial ordering of $\Kc$.
  Specifically, $c_{\Kc} \ge c_{\Kc'}$ if $\Kc \subseteq \Kc'$. 
\end{claim}
Note that this is intuitive since demanding more users to decode the
same {sub-message} should become more costly, and therefore higher-order
{sub-messages} have lower priority. 
With Algorithm~\ref{algo:SOA}, one can choose the $N$ ``best'' streams for
an arbitrary number $N \le 2^K-1$. Although it is a heuristic way to
identify a given number of best streams, the complexity is much lower
compared to the exact solution. 
Note that to find the exact solution, one would need 
to consider all subsets of the $2^K-1$ streams with cardinality
$N$, i.e., to check all $\binom{2^K-1}{N}$ possibilities. For instance,
there are $4\times 10^{14}$ possibilities when $K=8$ and $N=8$. For each
possibility, one needs to solve the sum rate maximization problem
subject to the previously derived rate constraints. Therefore, the complexity of
such optimal algorithms is prohibitive for practical application. On the
other hand, we can show that the complexity of Algorithm~\ref{algo:SOA}
is $O(4^K)$ which is around $7\times10^4$ and no optimization problem
needs to be solved. Indeed, in Algorithm~\ref{algo:SOA}, for each
$\Ic\subset[K]$, the complexity for computing
$\|\widetilde{\Wm}_{\!\Ic}\|^2$, (including finding $\Wm_{\Ic}$,
verifying~\eqref{eq:suff_cond}, and computing the norm)  is $O(K^2M + K
M^2 + M^2) = O(K M^2)$; the complexity for finding the minimum in
\eqref{eq:cK} is $O(2^K)$; the sorting has complexity $O(2^K \log(2^K))
= O(K 2^K)$. Therefore, the overall complexity of
Algorithm~\ref{algo:SOA} is  $O\left( KM^2 2^K + 2^K 2^K + K 2^K
\right) = O(4^K)$, assuming reasonably that $K^2 M \le 2^K$ when
$K$ and $M$ become large. 

To show that Algorithm~\ref{algo:SOA} can be practically effective, we
run a numerical simulation for $K=4$ users. We use the one-ring
scattering model~\cite{adhikary2013joint} to introduce spatial
correlation, in which scenario RS is particularly useful. In the
simulation, we consider two groups with low inter-group correlation and
high intra-group correlation. Each of the four users can be associated
randomly with one of the groups. We then apply Algorithm~\ref{algo:SOA}
to order the streams. In Fig.~\ref{fig:ordered_stream}, we show the
achievable rate when the $N$ ``best'' streams out of the total $2^K-1$ streams 
are activated. We also plot the achievable rate of the $1$-layer RS scheme
in which all private streams and one common stream to all users (the
stream $[K]$) are activated. We observe that when $N=K+1=5$, the algorithm
chooses a common stream to combine with the $K=4$ private streams, which
improves the sum rate performance. It outperforms the $1$-layer scheme
that does not depend on the channel realization. This example shows that our algorithm 
can provide an effective and efficient way to select a given number of streams adapted
to the channel condition.

 \begin{figure}
 	\centering
 	\includegraphics[width=0.48\textwidth]{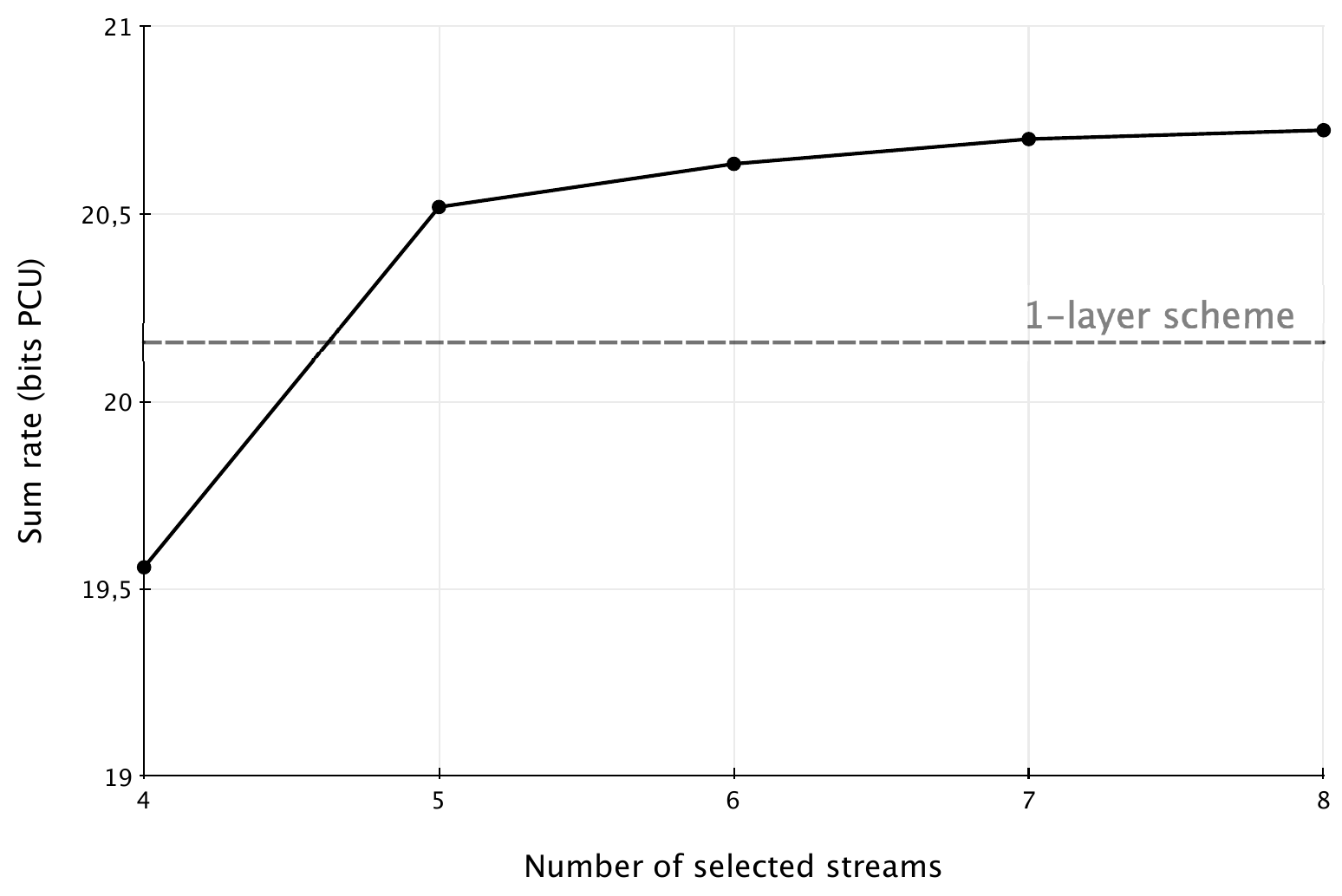}
 	\caption{Average sum rates vs.~the number of active streams selected by
        Algorithm~\ref{algo:SOA}. $K=4, M=4, G=2, P=30\,$dB. Averaged
        over $1000$ channel realizations generated by the one-ring
        scattering model~\cite{adhikary2013joint} with parameters
        $\Delta=\frac{40\pi}{180},
        \theta_g=-\frac{\pi}{3}+\Delta+\frac{\pi}{3}(g-1)$, $g=1,2$.}
        \label{fig:ordered_stream}
 \end{figure} 

\section{Constant-Gap Optimality and Non-Optimality}\label{sec:optimality}\label{sec:2U3U}

In the previous sections, we have investigated the achievable rate of 
linear precoded RS schemes. In this section, we are interested in the
optimality of such schemes as compared to the capacity region in the
constant-gap sense.  

\subsection{Linear precoding alone is not {constant-gap} optimal}\label{sec:LP}

We have shown in Section~\ref{sec:SISO} that 
even the single-user transmission, as an extreme case of the linear schemes, can achieve the sum capacity
to within a constant gap. {We shall now} show that the same
optimality does not hold with multiple antennas with linear precoding
alone. For this purpose, we consider a two-user MISO BC with {two} transmit
antennas. Note that, in this case, the channel matrix
$\Hm_k\in\mathbb{C}^{1\times2}$ is a row vector for each user $k,k=1,2$,
therefore, it is instead denoted by $\hv_k$ following our notational
convention. As discussed in Section \ref{sec:capacity_def}, it is
without loss of optimality to consider the following quantity as the sum
capacity as we are only interested in the capacity to within a constant gap
\begin{align}
C_{\text{sum}}&\approx\log \det(\Id + P \Hm^\H \Hm) \nonumber\\
& =\log \left(1 + P \|\hv_1\|^2 + P \| \hv_2\|^2 + P^2
\det(\Hm\Hm^\H)\right),
\end{align}
where $\Hm := [\hv_1^\T\ \hv_2^\T]^\T$ since $\hv_1$ and $\hv_2$ here
are row vectors. 
The above quantity is within $\log 3$ {bits/s/Hz} to
\begin{multline}
  \max\Bigl\{ \log(1+P \|\hv_1\|^2),
  \log(1+ P\|
  \hv_2\|^2), \\
  \log(1+P^2\det(\Hm\Hm^\H))
  \Bigr\}. \label{eq:DPC0}
\end{multline}
Note that the first two terms in \eqref{eq:DPC0} can be achieved with single-user transmission,
by serving the stronger user. Therefore, the only non-trivial case is when $\log(1+P^2\det(\Hm\Hm^\H))$
is the dominating term in \eqref{eq:DPC0}.

To prove our statement, let us assume that
the channel matrix has the following triangular form
\begin{align}
\Hm &= \begin{bmatrix} 1 & 0 \\ f & g \end{bmatrix}, \label{eq:triang_H}
\end{align}%
where the normalization can be done by scaling the transmit power; hence, $\hv_1 = [1\quad 0]$ and $\hv_2=[f\quad g]$.
In this case, the sum capacity \eqref{eq:DPC0} becomes
\begin{multline}
C_{\text{sum}} \approx \max\Bigl\{ \log(1+P), \\ 
\log(1+P |f|^2+P |g|^2), \log(1+P^2|g|^2)   \Bigr\}. \label{eq:DPC}
\end{multline}

Now let us restrict ourselves to linear precoding schemes at the transmitter and treating interference as noise at the receivers.
In particular, we let $\rvVec{X} = \rvVec{X}_1 + \rvVec{X}_2$ such that $\E[\rvVec{X}_1 \rvVec{X}_1^\H] = \Qm_1$ and $\E[\rvVec{X}_2 \rvVec{X}_2^\H] = \Qm_2$ with the following eigenvalue decompositions
\begin{align}
\Qm_1 &= \begin{bmatrix}  u_1 &v_1 \\ \tilde{v}_1 & \tilde{u}_1 \end{bmatrix} \begin{bmatrix}
\lambda_1 & \\ & \mu_1 \end{bmatrix} \begin{bmatrix}  u_1^* & \tilde{v}_1^* \\ v_1^* &
\tilde{u}_1^* \end{bmatrix}, \label{eq:AandB1}\\
\Qm_2 &= \begin{bmatrix}  u_2 &v_2 \\ \tilde{v}_2 & \tilde{u}_2 \end{bmatrix} \begin{bmatrix} \lambda_2 & \\ & \mu_2 \end{bmatrix} \begin{bmatrix}
 u_2^* & \tilde{v}_2^* \\ v_2^* & \tilde{u}_2^* \end{bmatrix},
\label{eq:AandB2}
\end{align}
where $|\tilde{u}_1|^2 = | u_1|^2 = 1-|\tilde{v}_1|^2 = 1 - |v_1|^2$ and $\lambda_1 \ge \mu_1 \ge 0$ without loss of generality; the same convention is
applied for $\Qm_2$. Due to the Gaussian signaling, we have
\begin{align}
R_1 &= \log\left( 1 + \frac{\hv_1\Qm_1 \hv_1^\H}{1 +  \hv_1 \Qm_2
\hv_1^\H} \right)\\
&= \log\left( 1 + \frac{\Qm_1(1,1)}{1 + \Qm_2(1,1)} \right), \label{eq:R1_LP}\\
R_2 &= \log\left( 1 + \frac{ \hv_2 \Qm_2 \hv_2^\H}{1 +\hv_2 \Qm_1
\hv_2^\H} \right).
\label{eq:R2_LP}
\end{align}%
Note that we are only interested in the case with
\begin{align}
\frac{\Qm_1(1,1)}{1 +  \Qm_2(1,1)} \ge 1 \quad \text{and} \quad \frac{\hv_2 \Qm_2 \hv_2^\H}{1 + \hv_2 \Qm_1 \hv_2^\H} \ge 1,
\label{eq:cond_AB}
\end{align}
for otherwise it is equivalent to the single-user case to within a constant gap. In this case, the achievable sum rate with linear precoding can be
written as
\begin{align}
\!\!\!\!\!\!R_1 + R_2 &\approx \log\left(\frac{\Qm_1(1,1)}{1 +  \Qm_2(1,1)} \right) + \log\left( \frac{\hv_2 \Qm_2 \hv_2^\H}{1 + \hv_2 \Qm_1 \hv_2^\H} \right) \\
\!\!\!\!\!\!&= \log\left(\frac{\Qm_1(1,1)}{1 + \hv_2 \Qm_1 \hv_2^\H} \right) + \log\left( \frac{\hv_2 \Qm_2 \hv_2^\H}{1 +  \Qm_2(1,1)} \right)\!. \label{eq:sumrate-LP}
\end{align}%
We can now maximize over $\Qm_1$ and over $\Qm_2$ separately. In fact, one can show the following lemma.
\begin{lemma}\label{lemma:ub}
	For any $\Qm_1$ and $\Qm_2$ in \eqref{eq:AandB1} and
        \eqref{eq:AandB2}, we have
	\begin{align}
	\frac{\Qm_1(1,1)}{1 + \hv_2 \Qm_1 \hv_2^\H}
	&\le 2 \min\left\{ \frac{2}{|f|^2}+2\frac{|g|^2}{|f|^2}\lambda_1,\ \lambda_1 \right\}, \label{eq:ub_A}\\
	\frac{\hv_2 \Qm_2 \hv_2^\H}{1 +  \Qm_2(1,1)} &\le 2 |f|^2 + 2 |g|^2 \lambda_2.\label{eq:ub_B}
	\end{align}%
\end{lemma}
\begin{proof}
  See Appendix~\ref{app:lemma3}. 
\end{proof}

To show that linear precoding is not constant-gap optimal,
  we consider high SNR $P$ and let the channel coefficients scale with $P$ as $f = P^{\alpha_f}$ and
	$g = P^{\alpha_g}$ for some $\alpha_f,\alpha_g\in\mathbb{R}$. It follows that the achievable sum rate also scales with $P$ as
	$d_{\text{LP}}(\alpha_f, \alpha_g) \log P + O(1)$, while the sum
        capacity scales as $d_{\text{DPC}}(\alpha_f, \alpha_g) \log P +
        O(1)$. Here, the pre-log factor is {the} GDoF as explained in
        Section~\ref{sec:def_constant}. {We shall show that there
        exist some $(\alpha_f,\alpha_g)$ such that}
        $d_{\text{LP}}(\alpha_f, \alpha_g) < d_{\text{DPC}}(\alpha_f,
        \alpha_g)$. 
	
        Indeed, when $\alpha_f>\alpha_g>\alpha_f-\frac{1}{2}\ge 0$, \eqref{eq:ub_A} scales as $P^{1+2\alpha_g-2\alpha_f}$ and \eqref{eq:ub_B} scales as $P^{1+2\alpha_g}$. It follows that $d_{\text{LP}} \le 2 + 4\alpha_g - 2\alpha_f$. From \eqref{eq:DPC}, we verify that
	$d_{\text{DPC}} = \max\left\{ 1,\ 1+2\alpha_f,\ 1+2\alpha_g,\
        2+2\alpha_g \right\} = 2 + 2\alpha_g$. Thus, we have shown that
        $d_{\text{DPC}} > d_{\text{LP}}$ for such $(\alpha_f,\alpha_g)$.
        Hence, linear precoding is not GDoF optimal, thus not
        constant-gap optimal.

\begin{remark}\label{remark:nonGaussian}
It is important to emphasize that the above results are based on the
assumption of Gaussian signaling. In fact, Gaussian input has been
proved to be strictly suboptimal in some multi-user settings. For
instance, in~\cite{dytso2016interference}, the authors have investigated
the two-user Gaussian interference channel with point-to-point codes,
and showed that a mixed input is needed to achieve the optimal GDoF. 
There, the mixed input is the sum of a discrete
random variable and a Gaussian variable. With the mixed input, the
optimal decoding, e.g., maximum likelihood decoding, exploits the
structure of the interference and achieves a better performance than in
the case with Gaussian interference. Essentially, as the authors of
\cite{dytso2016interference} pointed out, the discrete part carries
somehow a sort of ``common information'' that both receivers can
exploit. That explains why RS is not needed with such inputs to
achieve the optimal GDoF. Note, however, that the optimal decoding in this case may be much more 
involved than the one for Gaussian interference. The latter only needs a
simple nearest neighbour decoding. 
\end{remark}

In the following, we consider linear precoding schemes with rate-splitting.

\subsection{Linear precoded RS is constant-gap optimal with two
users}\label{sec:Two_Op}

The rate region of the two-user BC with RS is given in
Example~\ref{ex:2user} from \eqref{eq:ex1} to \eqref{eq:ex4}. Defining
\begin{gather}
  R_1 = \tilde{R}_1 + \tilde{R}_{12}^{(1)}, \quad
  R_2 = \tilde{R}_2 + \tilde{R}_{12}^{(2)}, 
\end{gather}
and applying the Fourier-Motzkin elimination~\cite{NIT}, we obtain the following achievable region
\begin{gather}
  R_1 \le C_1, \quad R_2 \le C_2,\\
  R_1 + R_2 \le C_{12},
\end{gather}
which corresponds to the capacity region $\mathcal{C}_{\text{MAC}}(\{\Hm_k^\H\}_k, P\Id)$ of the dual MAC. We thereby
establish the constant-gap optimality of the proposed RS scheme with
MMSE precoding {in the two-user case}.

\subsection{Linear precoded RS is constant-gap sub-optimal with
three users}

We shall show that the constang-gap optimality does not extend
beyond two users. To that end, we first present the constant-gap sum
rate of the three-user case.
\begin{proposition} \label{prop:3-user-sr}
	The optimal sum rate $R_{\text{sum}}$ of the proposed RS scheme
        with MMSE precoding in the three-user case is within a constant
        gap to
	\begin{multline}
		R_{\text{sum}}^* :=
		\max\Biggl\{C_{12},C_{13},C_{23}, \\
		\min\biggl\{C_{123}, 
                \displaystyle\min_{k=1,2,3}\frac{{l_k^{(2)}-C_k}}{2}
                +\xi \biggr\}\Biggr\} ,
		\label{eq:three_sr}
	\end{multline}
        where $l_i^{(k)}$, $i,k\in[K]$, is defined by \eqref{eq:def_l},
        and
        \begin{equation}
                \xi:={\frac{C_1+C_2+C_3-C_{12}
                -C_{23}-C_{13}+3C_{123}}{2}}. 
        \end{equation}%
\end{proposition}
\begin{proof}
  This is a direct consequence of Proposition~\ref{prop:sumrate}.
  Indeed, if only two users out of the three, say, users $1$ and
  $2$, are activated, then the sum rate $C_{12}$ is achievable to within
  a constant gap according to Proposition~\ref{prop:sumrate}. Similarly, 
  $C_{13}$ and $C_{23}$ can be achieved if we activate 
  another subset instead. If all three users are activated, then the
  achievable constant-gap sum rate in Proposition~\ref{prop:sumrate} becomes the second term inside the
  $\max\{\cdot\}$ in \eqref{eq:three_sr}. Note that activating only one
  user achieves $\max\{C_1, C_2, C_3\}$ that is strictly smaller than
  $\max\{C_{12},C_{13},C_{23}\}$. 
\end{proof}

From the above result, we can prove the constant-gap sub-optimality of the proposed
RS scheme. 
\begin{corollary}
  The proposed scheme is not GDoF optimal~(and therefore not constant-gap
  optimal) in the three-user case.
\end{corollary}
\begin{proof}
In order to prove the suboptimality, it is enough to find a class of
channel matrices $\Hm$ such that $C_{123} - {R}^*_{\text{sum}}$
can be arbitrarily large, where $C_{123}$ is the sum capacity of the
channel. Since ${R}^*_{\text{sum}}$ in \eqref{eq:three_sr} is still
quite involved due to the presence of {$l_k^{(2)}$}, $k=1,2,3$, we further
upper bound ${R}^*_{\text{sum}}$ using the following inequality. 
\begin{equation}
  {l_k^{(2)}} \le {\log\det(\Id + 2 P \Hm_k\Hm_k^\H)
  \approx{}} C_k. \label{eq:tau<C}
\end{equation}
Hence, we have  
\begin{multline}
  {R}^*_{\text{sum}} \lessapprox \overline{R}^*_{\text{sum}} :=\max\biggl\{C_{12}, C_{13}, C_{23},
 \\ {\frac{{{C_1} + {C_2} + {C_3} -
 {C_{12}} - {C_{23}} - {C_{13}}} + {3C_{123}}  }{2}}\biggr\}.
 \label{eq:tmp01}
 \end{multline}
In the following, we shall show that $\overline{R}^*_{\text{sum}}$ can
be arbitrarily smaller than $C_{123}$. To that end, we shall focus on the high
SNR regime and look at the pre-log of the rate expressions. Let us consider a channel with
\begin{equation}
  \Hm_{\!1} = [1\ 0\ 0],\
  \Hm_{\!2} = [0\ 1\ 0],\ \Hm_{\!3} = [P^{\frac{\alpha}{2}}\
  P^{\frac{\alpha}{2}}\ 1].
\end{equation}%
Define the pre-log $d_{\mathcal{K}} :=
\lim_{P\to\infty} \frac{C_{\mathcal{K}}}{{\log P}}$, and we have
\begin{gather}
  d_1 = d_2 = 1, \quad d_3 = 1+\alpha, \\
  d_{12} = 2,\quad d_{13} = d_{23} = 2 + \alpha, \quad d_{123} = 3.
\end{gather}
From \eqref{eq:tmp01}, we have the following upper bound for the pre-log
of ${R}^*_{\text{sum}}$, 
\begin{equation}
  \max\Bigl\{2, 2+\alpha, 3 - \frac{\alpha}{2} \Bigr\},
\end{equation}
which is strictly smaller than the optimal sum GDoF $d_{123} = 3$ for
any $0 < \alpha < 1$. This implies the constant-gap sub-optimality of the
proposed RS scheme.
\begin{figure*}[t]
	\centering
	\includegraphics[width=0.8\textwidth]{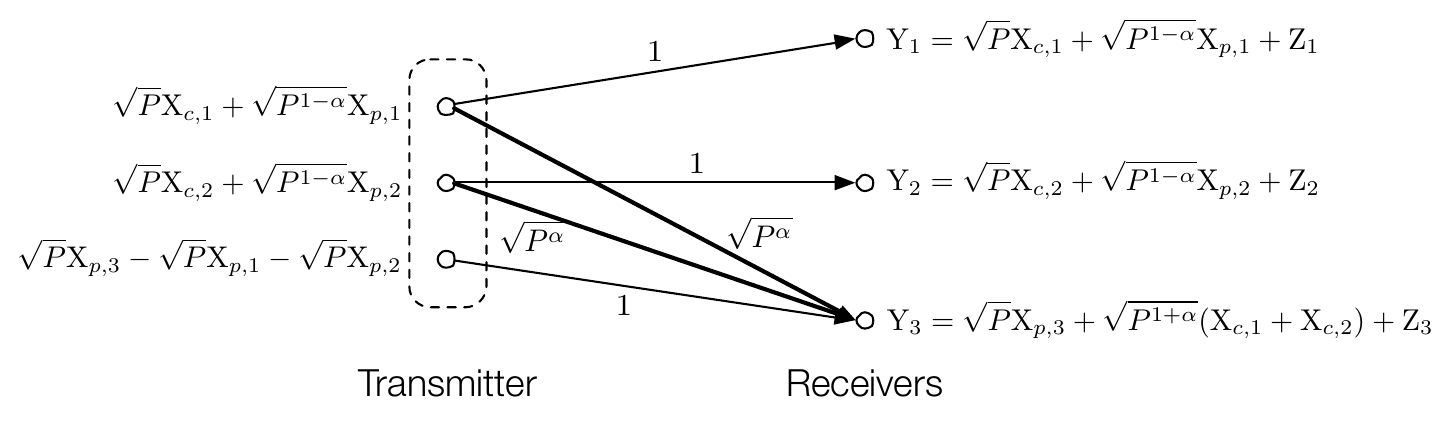}
	\caption{A pathological three-user MISO BC example.}\label{fig:example}
\end{figure*} 
\end{proof}

\subsection{Deficiency of receiver-side interference mitigation}
One may wonder why the proposed RS scheme is constant-gap optimal 
in the two-user case but not in the three-user case. In particular, is
it possible to improve the current RS scheme with a better
precoding~(other than the MMSE precoding) or with a more
sophisticated decoding scheme? 
To have a better understanding of why the RS scheme fails in the
three-user case, let us have a closer look at the above pathological
example. From the dual MAC, we know that a GDoF triple~$(1,1,1)$ is
achievable, e.g., with joint decoding or successive
interference cancellation in the uplink receiver. Specifically,
the receiver can first decode user~3's message using only the third
antenna, obtaining GDoF~$1$, and remove it before decoding user~1 and
user~2's messages from the first and the second antennas, respectively.
In the downlink, with DPC, the exact reverse procedure can
be applied and the same GDoF triple can be obtained. This is the
advantage of transmitter-side interference cancellation where the
transmitter manipulates optimally all the signals so that the
interference at the receivers' side is minimized.

With the RS scheme, however, the receivers are
interference-limited. To see this, let us impose that user~1 and user~2
both have GDoF~$1$. Thus, full power $P$ must be used for antennas
$1$ and $2$ to send the users' signals, which generates an interference
power $P^{1+\alpha}$ at user~3. Note that user~3's signal, in order not to
interfere with user 1 and 2's signals, must be essentially sent from antenna
$3$, arriving at user~3 with power $P$. Unless the interference could be
fully cancelled or decoded and removed, full GDoF~$1$ would not be achievable.
As shown in Figure~\ref{fig:example}, we can split the signal~1 into common and private parts
$\sqrt{P}\rv{X}_{c,1} + \sqrt{P^{1-\alpha}}\rv{X}_{p,1}$ with DoF
$\alpha$ and $1-\alpha$, respectively. Similarly
for signal~$2$, we use $\sqrt{P}\rv{X}_{c,2} + \sqrt{P^{1-\alpha}}\rv{X}_{p,2}$.
Signal~$3$ carries the private information for user~3 and cancels the
private parts in signal~$1$ and $2$, namely,
$\sqrt{P}\rv{X}_{p,3} - \sqrt{P}\rv{X}_{p,1} - \sqrt{P}\rv{X}_{p,2}$, so
that user~3 receives
$\sqrt{P}\rv{X}_{p,3} + \sqrt{P^{1+\alpha}}(\rv{X}_{c,1} + \rv{X}_{c,2}) + \rv{Z}_3$.
Note that $\rv{X}_{p,3}, \rv{X}_{c,1}, \rv{X}_{c,2}$, with a total DoF
$1+2\alpha$, must be decoded by user~$3$ in order to recover the private
DoF of $1$. This is impossible since the maximum GDoF for receiver~3 is
$1+\alpha$. Instead, user~3 can only achieve a GDoF of
$1+\alpha-2\alpha = 1-\alpha$. In other words, the RS scheme achieves
the $(1,1,1-\alpha)$ GDoF triple instead of $(1,1,1)$. Note that the
above discussion is independent of the precoding scheme and the decoding
scheme, which implies that the sub-optimality of the RS scheme cannot be
resolved in these directions.  

In fact, the fundamental issue of the RS scheme in the above example is
that independent codebooks are used for different streams. Intuitively,
the interference signal space becomes too large for any individual
receiver. If one could align different interferers into a reduced
subspace, however, then the achievable rate could be improved. In
particular, in the above case, if the information in $\rv{X}_{c,1} +
\rv{X}_{c,2}$ only occupies a DoF of $\alpha$ instead of $2\alpha$, then
user~3 could decode the sum of the interferences instead of the
individual interferences, and achieves the GDoF $1+\alpha-\alpha = 1$.
This is precisely the idea of interference alignment
\cite{cadambe2008interference, motahari2014real}. Instead of using
independent codebooks, one could use the same lattice codebook for
$\rv{X}_{c,1}$ and $\rv{X}_{c,2}$ in such a way that the sum is still
within the same codebook and thus have a reduced rate. Therefore,
combining RS and interference alignment, it is possible to reduce the
GDoF gap and may be possible to attain constant-gap optimality. 
One may also improve the performance by using
non-linear precoding for interference cancellation. For example, a
recent work \cite{flores2018tomlinson} proposes a RS scheme with
Tomlinson-Harashima precoding which has lower implementation complexity
than the DPC and is shown to outperform the linear precoding schemes.
Nevertheless, such improvements come at the price of a higher complexity
at the transmitter side, which limits the practical and theoretical
interests. 

Finally, it is also worth mentioning that Gaussian signaling is known to achieve
the capacity region of a two-user BC with common
message~\cite{geng2014capacity}, while the optimal signaling for more
users with common messages is still unknown. The above rate analysis
on the three-user RS scheme, and thus the conclusion, may not hold with 
a different signaling.

\section{Conclusion}\label{sec:conclusion}

We have investigated the achievable rate region of linearly
precoded rate-splitting schemes in the $K$-user MIMO broadcast channel
to within a constant gap. In particular, we have derived the achievable
constant-gap sum rate for $K\le3$, and obtained closed-form upper
bounds for $K>3$. The constant-gap results, though asymptotic, provide
useful insights that guided us to propose a practical stream elimination
algorithm. Our analyses also revealed the constant-gap optimality of
linearly precoded RS with respect to the fundamental capacity region in
the two-user case. While such optimality does not extend beyond two
users, we have provided explanations on the deficiency and potential
remedies. The results presented in the initial version of this work have been
followed up in \cite{jsac2020} with additional precoder optimization and
numerical simulations. Therein, the stream elimination algorithm has
also been applied and shown effective in practical scenarios.  

Note that for $K$-user broadcast channels with general message sets ---
even the Gaussian MIMO case with degraded message sets --- the capacity
region is still unknown. In those cases, rate-splitting goes beyond
a method to simplify transmission, as in our case with linear
precoding, and becomes an essential tool to improve the achievable rate
region when combined with binning~\cite{romero2020rate}.

\appendix
\subsection{Proof of Lemma~\ref{lemma:2}}\label{app:lemma2}
Denote $P_i:=\sum_{k=1}^K\trace(\Qm_k^{(i)}), i=1,\cdots, N$, then the power constraint can be expressed as 
\begin{align}
  P_s:=\sum_{i=1}^N\lambda_iP_i\le P. \label{eq:tmp335}
\end{align}
Define the following variable $\mu_i$ for each resource portion $i$
\begin{equation}
\mu_i := \begin{cases}
1, &\text{if } P_i\le P_s,\\
\frac{P_s}{P_i}, & \text{otherwise.}
\end{cases} \label{eq:def_mu}
\end{equation}
One can notice that $\mu_i\le 1, \forall i$. With resource-sharing, user~$k$ achieves the following rate
  \begin{align}
    \MoveEqLeft{\sum_{i=1}^N \lambda_i \log\det\biggl( \Id +
    \Bigl(\Id+\sum_{l=1,\,l\ne k}\Hm_k
    \Qm_l^{(i)} \Hm_k^\H\Bigr)^{-1} \Hm_k \Qm_k^{(i)} \Hm_k^\H \biggr)
    } \nonumber \\
     &= \sum_{i=1}^N \lambda_i \log\det\biggl( \Id +
    \frac{1}{\mu_i}\sum_{l=1}^K \Hm_k (\mu_i\Qm_l^{(i)}) \Hm_k^\H
    \biggr) \nonumber \\ 
    &\quad- \sum_{i=1}^N \lambda_i \log\det\biggl( \Id +
    \frac{1}{\mu_i}\sum_{l=1,\,l\ne k}^K \Hm_k (\mu_i\Qm_l^{(i)})
    \Hm_k^\H \biggr) \nonumber \\
    &\le \nr \sum_{i=1}^N \lambda_i\log\frac{1}{\mu_i} \nonumber \\
    &\quad + \sum_{i=1}^N \lambda_i \log\det\biggl( \Id +
    \sum_{l=1}^K \Hm_k (\mu_i\Qm_l^{(i)}) \Hm_k^\H
    \biggr) \nonumber \\ 
    &\quad- \sum_{i=1}^N \lambda_i \log\det\biggl( \Id +
    \sum_{l=1,\,l\ne k}^K \Hm_k (\mu_i\Qm_l^{(i)})
    \Hm_k^\H \biggr) \label{eq:tmp823} \\
    &\le  \nr + \sum_{i=1}^N \lambda_i R'_{k,i}, \label{eq:tmp824}
   \end{align}
   where 
   \begin{IEEEeqnarray}{rCl}
   R'_{k,i} &:=& \log\det\biggl( \Id +
   \sum_{l=1}^K \Hm_k (\mu_i\Qm_l^{(i)}) \Hm_k^\H
   \biggr) \nonumber\\
   &&  -\> \log\det\biggl( \Id + \sum_{l=1,\,l\ne k}^K \Hm_k (\mu_i\Qm_l^{(i)}) \Hm_k^\H \biggr). 
   \end{IEEEeqnarray}
   The inequality \eqref{eq:tmp823} is from
   $\mu_i \le 1$, $\forall\,i$; the last inequality is from the
   concavity of the $\log$ function, i.e., $ \sum_{i=1}^N
   \lambda_i\log\frac{1}{\mu_i}\le
   \log\sum_{i=1}^N\lambda_i\frac{1}{\mu_i}$, and the fact that
   \begin{align}
\sum_{i=1}^N\frac{\lambda_i}{\mu_i}=\sum_{i:P_i\le P_s}\lambda_i+\sum_{i:P_i> P_s}\frac{\lambda_iP_i}{P_s}\le 1+1=2.
   \end{align}
   From the definition of $\{\mu_i\}_i$ in \eqref{eq:def_mu}, we can
   verify that, for any $i$,
   \begin{align}
     \sum_{k=1}^K \trace(\mu_i \Qm_k^{(i)}) &= \mu_i \sum_{k=1}^K
   \trace(\Qm_k^{(i)}) \\
   &= \mu_i P_i \\
   &= \min\{P_i, P_s\} \label{eq:tmp336} \\
   &\le P, \label{eq:tmp334}
   \end{align}
   where \eqref{eq:tmp336} is from \eqref{eq:def_mu}; \eqref{eq:tmp334}
   is from \eqref{eq:tmp335}.  Hence, we have, for each
   $i$, $(R'_{1,i}, \ldots, R'_{K,i})\in
   \mathcal{C}_{\text{BC}}^{\text{LP}}(\{\Hm_k\}_k, P)$, and 
   \begin{align}
   \Bigl(\sum_{i=1}^N \lambda_i R'_{1,i}, \ldots, \sum_{i=1}^N \lambda_i
   R'_{K,i}\Bigr)\in
   \Conv\left\{\mathcal{C}_{\text{BC}}^{\text{LP}}(\{\Hm_k\}_k,
   P)\right\}. \nonumber
   \end{align}
   Finally, from \eqref{eq:tmp824}, the proof is complete.
   \qed
   \subsection{Proof of Lemma~\ref{lemma:ub}}\label{app:lemma3} 
   	Let us first consider the part with $\Qm_1$
	\begin{align}
	\!\!\!\!
	\!
	\frac{\Qm_1(1,1)}{1 + \hv_2 \Qm_1 \hv_2^\H} &= \frac{| u_1|^2 \lambda_1 + |v_1|^2 \mu_1}{|f  u_1 + g \tilde{v}_1|^2 \lambda_1 + |f v_1 + g \tilde{u}_1|^2
		\mu_1 + 1}.
	\end{align}
	Since $|f  u_1 + g \tilde{v}_1|^2 + |f v_1 + g \tilde{u}_1|^2 = |f|^2 + |g|^2$, we use $a \le b$ to denote the ordered version of
	$|f  u_1 + g \tilde{v}_1|^2$ and $|f v_1 + g \tilde{u}_1|^2$, then we have $a\ge(|f  u_1| - |g \tilde{v}_1|)^2$ and $2b\ge |f|^2 + |g|^2$. Using the fact that $\lambda_1\ge\mu_1$, we have the
	following upper bound
	\begin{align}
	\MoveEqLeft{\frac{\Qm_1(1,1)}{1 + \hv_2 \Qm_1\hv_2^\H}} \nonumber \\
	&\le \frac{| u_1|^2 \lambda_1 + |v_1|^2 \mu_1}{a \lambda_1 +  b \mu_1 + 1} \\
	&\le \frac{| u_1|^2 \lambda_1 }{a \lambda_1 + 1} + \frac{|v_1|^2 \mu_1}{b \mu_1 + 1} \\
	&\le \frac{| u_1|^2 \lambda_1 }{(|f  u_1| - |g \tilde{v}_1|)^2 \lambda_1 + 1} + \frac{\mu_1}{\frac{|f|^2 + |g|^2}{2} \mu_1 + 1} \\
	&\le \max_{ u_1:\, |f  u_1| \ge |g \tilde{v}_1|}
	\frac{| u_1|^2 \lambda_1 }{(|f  u_1| - |g \tilde{v}_1|)^2 \lambda_1 + 1} \nonumber
	\\
	&\qquad + \min\left\{ \left(\frac{|f|^2 + |g|^2}{2}\right)^{-1}\!\!, \ {\mu_1} \right\}  \label{eq:ddj}\\
	&\le \max_{ u_1:\, |f  u_1| \ge |g \tilde{v}_1|} \min\left\{ \frac{| u_1|^2}{(|f  u_1| - |g
		\tilde{v}_1|)^2}, \ {| u_1|^2 \lambda_1 }\right\} \nonumber \\
	&\qquad +
	\min\left\{ \left(\frac{|f|^2 + |g|^2}{2}\right)^{-1}, \ {\mu_1} \right\}  \label{eq:dds}\\
	&\le \min\left\{ \left( \frac{1}{\sqrt{\lambda_1}} + |g| \right)^2
	\frac{\lambda_1}{|f|^2},\ \lambda_1 \right\}  \nonumber \\
	&\qquad + \min\left\{ \left(\frac{|f|^2 + |g|^2}{2}\right)^{-1}, \ {\mu_1} \right\} \label{eq:ddk} \\
	&\le \min\left\{ \frac{2}{|f|^2}+2\frac{|g|^2}{|f|^2}\lambda_1,\ \lambda_1 \right\}
	\nonumber \\
	&\qquad + \min\left\{ \left(\frac{|f|^2 + |g|^2}{2}\right)^{-1}\!\!,
	\ {\mu_1} \right\} \\
	&\le 2 \min\left\{ \frac{2}{|f|^2}+2\frac{|g|^2}{|f|^2}\lambda_1,\ \lambda_1 \right\}, 
	\end{align}
	where \eqref{eq:ddj} is from the fact that the objective function is increasing with $|f  u_1|$ when
	$|f  u_1| \le |g \tilde{v}_1|$; since in \eqref{eq:dds}, $\frac{| u_1|^2}{(|f  u_1| - |g
		\tilde{v}_1|)^2}$ is decreasing with $| u_1|$ and ${| u_1|^2 \lambda_1 }$ is increasing with $| u_1|$,
	the max-min is attained when both terms are equalized or when $| u_1| = 1$;
	\eqref{eq:ddk} is indeed an upper bound of \eqref{eq:dds}. The second part can be shown as follows.
        \begin{IEEEeqnarray}{rCl}
          \IEEEeqnarraymulticol{3}{l}{\frac{\hv_2 \Qm_2 \hv_2^\H}{1 +
          \Qm_2(1,1)}}\nonumber \\ &=& \frac{|f  u_2 + g \tilde{v}_2|^2 \lambda_2 + |f v_2 + g \tilde{u}_2|^2 \mu_2}{| u_2|^2 \lambda_2 + |v_2|^2 \mu_2 + 1} \\
	&\le& 2  \frac{|f|^2 (| u_2|^2\lambda_2 + |{v}_2|^2 \mu_2) +
        |g|^2 (|\tilde{v}_2|^2 \lambda_2 + |\tilde{u}_2|^2 \mu_2)}{|
        u_2|^2 \lambda_2 + |v_2|^2 \mu_2 + 1} \IEEEeqnarraynumspace\label{eq:tmp777} \\
	&\le& 2 |f|^2 \frac{ | u_2|^2\lambda_2 + |{v}_2|^2 \mu_2}{| u_2|^2 \lambda_2 + |v_2|^2 \mu_2 +
		1}\nonumber \\
                &&+\> 2 |g|^2 \frac{ \lambda_2 }{| u_2|^2 \lambda_2 +
                |v_2|^2 \mu_2 + 1} \label{eq:tmp778} \\
	&\le& 2 |f|^2 + 2 |g|^2  \lambda_2,  
	\end{IEEEeqnarray}
        where to obtain \eqref{eq:tmp777} we used $|a+b|^2 \le 2|a|^2 +
        2 |b|^2$, $\forall\,a,b\in\mathbb{C}$; \eqref{eq:tmp778} is from
        the fact that $|\tilde{v}_2|^2 \lambda_2 + |\tilde{u}_2|^2 \mu_2
        \le |\tilde{v}_2|^2 \lambda_2 + |\tilde{u}_2|^2 \lambda_2 =
        \lambda_2$ due to $\mu_2\le\lambda_2$ and $|\tilde{u}_2|^2 +
        |\tilde{v}_2|^2 = 1$. 
        \qed

        \subsection{{Proof of the achievability of \eqref{eq:UB-Kuser} for
        $K=2$ and $K=3$}\label{app:prop5}}
First we present the following submodularity property that will be useful later. 
\begin{lemma}\label{lemma:submodularity}
	Let $\mathcal{A},\mathcal{A}',\mathcal{B}\subseteq [K]$ with
	$\mathcal{A}'\subseteq\mathcal{A}$. Then, we have
	\begin{equation}
	C_{\mathcal{A}\cup\mathcal{B}} - C_{\mathcal{A}} \le C_{\mathcal{A}'\cup\mathcal{B}} - C_{\mathcal{A}'}.\label{eq:property_redun}
	\end{equation}
\end{lemma}
\begin{proof}
	Indeed, this inequality can be proved directly using matrix properties
	or, more conveniently, with mutual information.
	Let $\mathcal{A} = \mathcal{A}'\cup\mathcal{A}''$ with
	$\mathcal{A}'\cap\mathcal{A}''=\emptyset$. Then, identifying
	$C_{\mathcal{A}\cup\mathcal{B}}$, $C_{\mathcal{A}}$,
	$C_{\mathcal{A}'\cup\mathcal{B}}$, and $C_{\mathcal{A}'}$ with
	$I(\rv{X}_{\mathcal{A}\cup\mathcal{B}}; \rv{Y})$,
	$I(\rv{X}_{\mathcal{A}}; \rv{Y} \cond \rv{X}_{\mathcal{B}})$,
	$I(\rv{X}_{\mathcal{A}'\cup\mathcal{B}}; \rv{Y} \cond
	\rv{X}_{\mathcal{A}''})$,
	$I(\rv{X}_{\mathcal{A}'}; \rv{Y} \cond \rv{X}_{\mathcal{B}}$, and
	$\rv{X}_{\mathcal{A}''})$, respectively, the above inequality becomes
	\begin{equation}
	I(\rv{X}_{\mathcal{B}}; \rv{Y}) \le
	I(\rv{X}_{\mathcal{B}}; \rv{Y} \cond \rv{X}_{\mathcal{A}''})
	\end{equation}
	after applying the chain rule of mutual information. This holds since
	$I(\rv{X}_{\mathcal{B}}; \rv{Y} \cond \rv{X}_{\mathcal{A}''}) =
	I(\rv{X}_{\mathcal{B}}; \rv{Y}, \rv{X}_{\mathcal{A}''}) \ge
	I(\rv{X}_{\mathcal{B}}; \rv{Y})$ due to the independence of
	$\rv{X}_{\mathcal{B}}$ and $\rv{X}_{\mathcal{A}''}$.
\end{proof}

\subsubsection{{The two-user case ($K=2$)}}
{When $K=2$, the rate region is given by
\eqref{eq:ex1}-\eqref{eq:ex4}. Then, we can verify that any rate
quadruple such that $\tilde{R}_1 = l_1^{\{1\}}, 
\tilde{R}_2 = l_2^{\{2\}},
\tilde{R}_{12}^{(1)} + \tilde{R}_{12}^{(2)} = \min\{ l_1^{\{1,2\}} -
l_1^{\{1\}}, l_2^{\{1,2\}} - l_2^{\{1\}} \}$ lies inside the region.
Indeed, one have $l_1^{\{1,2\}} - l_1^{\{1\}} \ge 0$ and
$l_2^{\{1,2\}} - l_2^{\{1\}} \ge 0$ due to the monotone property of
$l_k^{\collectionS}$ with respect to $\collectionS$. 
Hence, the sum rate given by the
  quadruple is $\tilde{R}_1 + 
  \tilde{R}_2 + \tilde{R}_{12}^{(1)} + \tilde{R}_{12}^{(2)} =
  l_1^{\{1\}} + l_2^{\{2\}}
  + \min\{ l_1^{\{1,2\}} - l_1^{\{1\}}, l_2^{\{1,2\}} - l_2^{\{1\}}\}$ which
  coincides with the upper bound \eqref{eq:UB-Kuser} for $K=2$ using the
  definition \eqref{eq:def_l}. 
  }

\subsubsection{{The three-user case ($K=3$)}}

When $K=3$, the possible minimal collections
$\underline{\collectionS}_1$ are $\{\{1\}\}$, $\{\{1,2\}\}$,
$\{\{1,3\}\}$, $\{\{1,2\}, \{1,3\}\}$, $\{\{1,2,3\}\}$, leading to
the following rate constraints from Proposition~\ref{pro:rate_region},
\begin{IEEEeqnarray}{rCl}
  \tilde{R}_1 &\le& {l_1^{ \{1\} }  = l_1^{(1)} 
  ={}} C_{123}-C_{23}, \IEEEeqnarraynumspace
\label{eq:con_1.1}\\
\tilde{R}_1+\tilde{R}_{12}\! &\le& {l_1^{\{1,2\}} 
={} } C_{13}-C_3
\label{eq:con_1.2}\\
\tilde{R}_1 + \tilde{R}_{13}\!&\le& {l_1^{\{1,3\}} 
={}} C_{12}-C_2,
\label{eq:con_1.3} \\
\tilde{R}_1+ \tilde{R}_{12} + \tilde{R}_{13}&\le& {l_1^{\{1,2\},\{1,3\}}
={} l_1^{(2)}},%
\label{eq:con_1.4}\\
\tilde{R}_1+ \tilde{R}_{12} + \tilde{R}_{13} + \tilde{R}_{123}&\le&
l_1^{\{1,2,3\}} = l_1^{(3)}
={} C_1,
\label{eq:con_1.5}
\end{IEEEeqnarray}%
where we recall that for each collection $\collectionS$, $l_k^{\collectionS} := \log\det\left( \Id +
\Hm_k \Qm_{\collectionS} \Hm_k^\H \right)$ with $\Qm_{\collectionS}$ being defined as
in \eqref{eq:def-Qs}, while for each $m\in[K]$, $l_k^{(m)}$ is defined
as in \eqref{eq:def_l}; slightly abusing the notation, $C_{ijk}$ denotes $C_{\{i,j,k\}}$ as defined in
\eqref{eq:def-C}; the above relationships between the $l$'s and
$C$'s can be verified by their definitions.  

Similarly, we can obtain the following constraints on the rates of the
re-assembled messages which should be decoded by receivers $2$ and $3$,
\begin{IEEEeqnarray}{rCl}
  \tilde{R}_2 &\le& {l_2^{ \{2\} }  = l_2^{(1)} ={}} C_{123}-C_{13},
\label{eq:con_2.1}\IEEEeqnarraynumspace\\
\tilde{R}_2 + \tilde{R}_{12}&\le& { l_2^{\{1,2\}}={}} C_{23}-C_3,
\label{eq:con_2.2}\\
\tilde{R}_2 + \tilde{R}_{23}&\le& {l_2^{\{2,3\}}={}} C_{12}-C_1,
\label{eq:con_2.3}\\
\tilde{R}_2+ \tilde{R}_{12} + \tilde{R}_{23} &\le&  {l_2^{\{1,2\},\{2,3\}}
= l_2^{(2)}}
\label{eq:con_2.4}\\
\tilde{R}_2+ \tilde{R}_{12} + \tilde{R}_{23}+ \tilde{R}_{123}&\le& {
l_2^{\{1,2,3\}} = l_2^{(3)}={}}C_2,
\label{eq:con_2.5}\\
\tilde{R}_3 &\le& { l_3^{ \{3\} }  = l_3^{(1)}={}} C_{123}-C_{12},
\label{eq:con_3.1}\\
\tilde{R}_3 + \tilde{R}_{13} &\le& { l_3^{\{1,3\}}={}} C_{23}-C_2,
\label{eq:con_3.2}\\
\tilde{R}_3 + \tilde{R}_{23} &\le&  { l_3^{\{2,3\}}={}} C_{13}-C_1,
\label{eq:con_3.3}\\
\tilde{R}_3+ \tilde{R}_{13} + \tilde{R}_{23} &\le&
{l_3^{\{1,3\},\{2,3\}} = l_3^{(2)}},
        \label{eq:con_3.4}\\
\tilde{R}_3+ \tilde{R}_{13} + \tilde{R}_{23}+ \tilde{R}_{123}&\le&
{ l_3^{\{1,2,3\}} = l_3^{(3)}={}} C_3.
\label{eq:con_3.5}
\end{IEEEeqnarray}

In the following, we shall show that there exists a rate tuple
satisfying all the constraints \eqref{eq:con_1.1}-\eqref{eq:con_3.5} that
achieve the sum rate upper bound~\eqref{eq:UB-Kuser}. Note that when
$K=3$, the second term in the bound~\eqref{eq:UB-Kuser} becomes
\begin{multline}
  \frac{1}{2}\Bigl( l_1^{(3)} + l_2^{(3)} + l_3^{(3)} + l_1^{(1)} +
  l_2^{(1)} + l_3^{(1)} \\
  + \min\left\{ l_1^{(2)} - l_1^{(3)}, l_2^{(2)} -
  l_2^{(3)}, l_3^{(2)} - l_3^{(3)} \right\} \Bigr). \label{eq:tmp654}
\end{multline}%

Without loss of generality we assume that the receivers are ordered such that
{
\begin{equation}
l_1^{(2)} - l_1^{(3)} \le l_2^{(2)} - l_2^{(3)}, \quad
l_1^{(2)} - l_1^{(3)} \le l_3^{(2)} - l_3^{(3)}. \label{eq:tmp673}
\end{equation}
The upper bound \eqref{eq:tmp654} can be further simplified to  
\begin{IEEEeqnarray}{rCl}
  \IEEEeqnarraymulticol{3}{l}{
  \frac{1}{2}\left( l_2^{(3)} + l_3^{(3)} + l_1^{(1)} +
  l_2^{(1)} + l_3^{(1)} + l_1^{(2)} \right)} \nonumber \\
  &=& {\frac{l_1^{(2)}+C_2+C_3+3C_{123}-C_{12} -C_{23}-C_{13}}{2}}. \label{eq:Rsum}
\end{IEEEeqnarray}

To prove the achievability of \eqref{eq:UB-Kuser}, let us consider the
following two cases. } 
\paragraph{Case $C_{123} < {\frac{{l_1^{(2)}}+C_2+C_3+3C_{123}-C_{12} -C_{23}-C_{13}}{2}}$}
The condition is equivalent to 
\begin{align}
  {l_1^{(2)}}+C_2+C_{3}+C_{123} > C_{13}+C_{23}+C_{12}. \label{eq:assump0}
\end{align}
In this case, we let
\begin{align}
  \tilde{R}_1&=C_{123}-C_{23}, \label{eq:tmp1111}\\
\tilde{R}_2&=C_{123}-C_{13},\\
\tilde{R}_3&=C_{123}-C_{12},\label{eq:tmp2222}\\
\tilde{R}_{12}&=C_{23}-C_{3}-(C_{123}-C_{13}),\label{eq:tmp222}\\
\tilde{R}_{13}&=C_{12}-C_{2}-(C_{123}-C_{23}),\label{eq:tmp333}\\
\tilde{R}_{23}&=C_{13}-C_{1}-(C_{123}-C_{12}),\label{eq:tmp444}\\
\tilde{R}_{123}&=C_1+C_2+C_3+C_{123}-C_{12}-C_{13}-C_{23}.  \label{eq:tmp888}
\end{align}
{First, we verify that the above rate tuple is non-negative. Indeed,
\eqref{eq:tmp1111}-\eqref{eq:tmp2222} are non-negative by the definition
of the $C$'s; from Lemma~\ref{lemma:submodularity},
\eqref{eq:tmp222}-\eqref{eq:tmp444} are non-negative, too; since
$l_1^{(2)} := \log\det\left( \Id + \Hm_1 \Qm_{\{\{1,2\},\{1,3\}\}}
\Hm_1^\H \right)\le \log\det\left( \Id + 2P\Hm_1 \Hm_1^\H \right)
\approx C_1$, the assumption \eqref{eq:assump0} implies the non-negativity of
\eqref{eq:tmp888} up to a constant gap.
Next,} 
we can verify that all constraints \eqref{eq:con_1.1} to
\eqref{eq:con_3.5}, except for \eqref{eq:con_1.4}, \eqref{eq:con_2.4},
and \eqref{eq:con_3.4}, are satisfied with equality. Then, we can verify
that plugging \eqref{eq:tmp1111}, \eqref{eq:tmp222}, and \eqref{eq:tmp333}
into \eqref{eq:con_1.4}, the constraint coincides with the assumption~\eqref{eq:assump0}.  
Similarly, both \eqref{eq:con_2.4} and \eqref{eq:con_3.4} are also
equivalent to \eqref{eq:assump0}.

\paragraph{Case $C_{123} \ge
{\frac{{l_1^{(2)}}+C_2+C_3+3C_{123}-C_{12} -C_{23}-C_{13}}{2}}$}
The condition is equivalent to
\begin{align}
  {l_1^{(2)}}+C_2+C_{3}+C_{123} \le C_{13}+C_{23}+C_{12}. \label{eq:assump1}
\end{align}
In this case, we let the six constraints \eqref{eq:con_1.1},  \eqref{eq:con_1.4}, \eqref{eq:con_2.1},
\eqref{eq:con_2.5}, \eqref{eq:con_3.1}, and \eqref{eq:con_3.5} be
satisfied with equality. It is equivalent to {having} the following
six equalities:
\begin{align}
  \tilde{R}_1&=C_{123}-C_{23}, \label{eq:tmp881}\\
	\tilde{R}_2&=C_{123}-C_{13},\\
	\tilde{R}_3&=C_{123}-C_{12},\label{eq:tmp883}\\
        \tilde{R}_{12}&=\frac{ {l_1^{(2)}} +C_2+C_{13}+C_{23}-C_3-C_{12}-C_{123}}{2},\label{eq:tmp884}\\
        \tilde{R}_{13}&=\frac{ {l_1^{(2)}} +C_3+C_{12}+C_{23}-C_2-C_{13}-C_{123}}{2},\label{eq:tmp885}\\
        \tilde{R}_{23}+\tilde{R}_{123}&=\frac{C_2+C_3+C_{12}+C_{13}-
        {l_1^{(2)}} -C_{23}-C_{123}}{2},
        \label{eq:sum23_123}
\end{align}
where, as in the previous case, \eqref{eq:tmp881}-\eqref{eq:tmp883} are non-negative by the definition
of the $C$'s; since $l_1^{(2)}=l_1^{ \{1,2\}, \{1,3\}}\ge \max\left\{
l_1^{\{1,2\}}, l_1^{\{1,3\}} \right\} = \max\left\{ C_{13}-C_3,
C_{12}-C_2 \right\}$, in \eqref{eq:tmp884} $\tilde{R}_{12}\ge
\frac{1}{2}(C_{13}-C_3 + C_{23} - C_{123}) \ge 0$ according to
Lemma~\ref{lemma:submodularity}, and similary for $\tilde{R}_{13}$; from
the assumption~\eqref{eq:assump1}, we can verify that in
\eqref{eq:sum23_123} $\tilde{R}_{23}+\tilde{R}_{123} \ge
C_2+C_3-C_{23}$ and the latter is non-negative according to
Lemma~\ref{lemma:submodularity}. 

In the following, we shall show that all the remaining constaints
are satisfied. First, \eqref{eq:con_1.2} can be rewritten as
\begin{align}
	\tilde{R}_1+\tilde{R}_{12}&=\frac{{l_1^{(2)}}+C_2+C_{13}+C_{123}-C_3-C_{23}-C_{12}}{2}\nonumber\\
        &\le C_{13}-C_3,
\end{align}
which is equivalent to the assumption \eqref{eq:assump1}. 
Similarly, we can verify that \eqref{eq:con_1.3},  \eqref{eq:con_2.2}, and \eqref{eq:con_3.2}
are all equivalent to the assumption \eqref{eq:assump1}. 

Then, we shall verify the constraints \eqref{eq:con_1.5},
\eqref{eq:con_2.3}, \eqref{eq:con_2.4}, \eqref{eq:con_3.3},
\eqref{eq:con_3.4}, which all involve $\tilde{R}_{23}$ and
$\tilde{R}_{123}$, are satisfied. 
Note that $\tilde{R}_{23}$ and $\tilde{R}_{123}$ remain undetermined
 except for their sum given by \eqref{eq:sum23_123}. 
Due to the symmetry of the constraints and assumptions on receivers~2
and 3, we only need to consider \eqref{eq:con_1.5}, \eqref{eq:con_2.3},
and \eqref{eq:con_2.4}. The three constraints, combined with
\eqref{eq:tmp881}-\eqref{eq:tmp885}, can be rewritten as
\begin{align}
  \tilde{R}_{123} &\le C_1 - {l_1^{(2)}}, \label{eq:tmp721}\\
  \tilde{R}_{23} &\le C_{12} + C_{13} - C_1 - C_{123}, \label{eq:tmp722}\\
  \tilde{R}_{23} &\le \frac{2 {l_2^{(2)}} - {l_1^{(2)}} - C_2 + C_3 + C_{12} +
  C_{13} - C_{23} - C_{123}}{2}. \label{eq:tmp723}
\end{align}
We need to show that there exists $\tilde{R}_{23}\ge0$ and
$\tilde{R}_{123}\ge0$ such that all the three above constraints and
\eqref{eq:sum23_123} can be satisfied simultaneously. It is enough to
show the following.
\begin{itemize}
  \item The right hand sides of \eqref{eq:tmp721}-\eqref{eq:tmp723} are
    all non-negative to within a constant gap. It can be verified,
    1)~for \eqref{eq:tmp721} {since $l_1^{(2)} \le \log\det(\Id + 2P \Hm
    \Hm^\H) \approx C_1$}; 2)~for \eqref{eq:tmp722} from
    Lemma~\ref{lemma:submodularity}; and 3)~for \eqref{eq:tmp723} we
    have
    \begin{IEEEeqnarray}{rCl}
      \IEEEeqnarraymulticol{3}{l}{2 {l_2^{(2)}} - {l_1^{(2)}} - C_2 + C_3 + C_{12} + C_{13} - C_{23} -
C_{123}}\nonumber \\
       &\ge& {l_2^{(2)}} - {l_1^{(2)}} - C_2 + C_{23} + C_{12} + C_{13} - C_{23} -
       C_{123} \nonumber \\
       &\ge&  - C_1   + C_{12} + C_{13}  - C_{123} \nonumber \\
       &\ge& 0,
    \end{IEEEeqnarray}
    where the first inequality is from 
    ${l_2^{(2)}} \ge C_{23} - C_3$ {due to Lemma~\ref{lemma:l-C}}; the second one is from the
    assumption~\eqref{eq:tmp673}; the last one is from
    Lemma~\ref{lemma:submodularity}. 
  \item The sum of the right-hand sides of \eqref{eq:tmp721} and
    \eqref{eq:tmp722} is larger than the right-hand side of
    \eqref{eq:sum23_123}. Indeed, we have  
    \begin{align}
      \MoveEqLeft{( C_1 - {l_1^{(2)}}) +  (C_{12} + C_{13} - C_1 - C_{123})}  
      \nonumber \\ \MoveEqLeft[0]{-\frac{{C_2+C_3+C_{12}+C_{13}} -{l_1^{(2)}}-C_{23}-C_{123}}{2}} \nonumber
      \\
      &= -\frac{{l_1^{(2)}} + C_2 + C_3 - C_{12} - C_{13} - C_{23} +
      C_{123}}{2} \nonumber\\
      &\ge 0, 
    \end{align}
    where the inequality is from the assumption~\eqref{eq:assump1}. 
  \item The sum of right-hand sides of \eqref{eq:tmp721} and
    \eqref{eq:tmp723} is larger than the right-hand side of
    \eqref{eq:sum23_123}. Indeed, we have 
    \begin{IEEEeqnarray}{rCl}
      \IEEEeqnarraymulticol{3}{l}{( C_1 - {l_1^{(2)}}) }\nonumber \\ 
      \IEEEeqnarraymulticol{3}{l}{+\frac{2{l_2^{(2)}} - {l_1^{(2)}} - C_2 + C_3 + C_{12} + C_{13} - C_{23} - C_{123}}{2}}\nonumber \\ 
      \IEEEeqnarraymulticol{3}{l}{- 
      \frac{C_2+C_3+C_{12}+C_{13}-{l_1^{(2)}}-C_{23}-C_{123}}{2}} \nonumber
      \\
      &=& C_1 - {l_1^{(2)}} + {l_2^{(2)}} - C_2 \nonumber\\
      &\ge& 0, \nonumber
    \end{IEEEeqnarray}
    where the inequality follows from the assumption \eqref{eq:tmp673}. 
\end{itemize}
The proof is thus complete. \qed

\begin{IEEEbiographynophoto}{Zheng Li}
received the B.E. degree in information engineering from Southeast University, China, in 2013, and the M.S. degrees in telecommunications from CentraleSup{\'e}lec, France, in 2015, as well as in telecommunications from Southeast University, China, in 2016. She received her Ph.D. degree in telecommunications from CentraleSup{\'e}lec, France, in 2020. She is currently a researcher at Orange Labs Networks, France. 
\end{IEEEbiographynophoto}

\begin{IEEEbiographynophoto}{Sheng Yang} (M'07)
received the B.E. degree in electrical engineering from Jiaotong University, Shanghai, China, in 2001, and both the engineer degree and the M.Sc. degree in electrical engineering from Telecom ParisTech, Paris, France, in 2004, respectively. In 2007, he obtained his Ph.D. from Université de Pierre et Marie Curie (Paris VI). From October 2007 to November 2008, he was with Motorola Research Center in Gif-sur-Yvette, France, as a senior staff research engineer. Since December 2008, he has joined CentraleSupélec, Paris-Saclay University, where he is currently a full professor. From April 2015, he also holds an honorary associate professorship in the department of electrical and electronic engineering of the University of Hong Kong (HKU). He received the 2015 IEEE ComSoc Young Researcher Award for the Europe, Middle East, and Africa Region (EMEA). He was an associate editor of the IEEE transactions on wireless communications from 2015 to 2020. He is currently an associate editor of the IEEE transactions on information theory. 
\end{IEEEbiographynophoto}

\begin{IEEEbiographynophoto}{Shlomo Shamai~(Shitz)}
(S'80-M'82-SM'88-F'94-LF'18)
is with the Department of Electrical Engineering,
Technion---Israel Institute of Technology, where he is a
Technion Distinguished Professor, and holds the William Fondiller
Chair of Telecommunications.
Dr. Shamai (Shitz) is an IEEE Life Fellow, an URSI Fellow, a member of the
Israeli Academy of Sciences and Humanities and a foreign member of the
US National Academy of Engineering. He is the recipient of the 2011
Claude E. Shannon Award, the 2014 Rothschild Prize in
Mathematics/Computer Sciences and Engineering and the
2017 IEEE Richard W. Hamming Medal. He is a co-recipient of the 2018 Third
Bell Labs Prize for Shaping the Future of Information and Communications
Technology.

He is the recipient of numerous technical and paper awards and recognitions
of the IEEE (Donald G. Fink Prize Paper Award), Information Theory,
Communications and Signal Processing Societies as well as EURASIP.
He is listed as a Highly Cited Researcher (Computer Science)
for the years 2013/4/5/6/7/8.
He has served as Associate Editor for the Shannon Theory of the IEEE
Transactions on Information Theory, and has also served twice on the
Board of Governors of the Information Theory Society.
He has also served on the Executive Editorial Board of the IEEE Transactions
on Information Theory, the IEEE Information Theory Society Nominations
and Appointments Committee and the IEEE Information Theory Society,
Shannon Award Committee.
\end{IEEEbiographynophoto}

\end{document}